\documentclass[%
reprint,
amsmath,amssymb,
aps,
pra,
nofootinbib
]{revtex4-2}

\usepackage{graphicx}
\usepackage{dcolumn}
\usepackage{bm}
\usepackage{lipsum}

\usepackage[utf8]{inputenc}
\usepackage{mathtools}
\usepackage{amsmath}
\usepackage{amssymb}
\usepackage{amsthm}
\usepackage{graphicx,graphics,caption,subcaption}
\usepackage{lipsum}
\usepackage{esvect}
\usepackage{float}
\usepackage{bbold}
\usepackage{appendix}
\usepackage{xcolor}
\usepackage{hyperref}

\hypersetup{
	colorlinks=true,                
	breaklinks=true,                
	urlcolor= black,                
	linkcolor= blue,                
	bookmarksopen=false,
	filecolor=black,
	citecolor=black,
	linkbordercolor=blue,
	citebordercolor={0 1 0}
	}

\newtheorem{theorem}{Theorem}
\newtheorem{defn}{Definition}
\newtheorem{lemma}{Lemma}

\AtBeginDocument{\hypersetup{pdfborder={0 0 1}}}
\begin{document}
	\preprint{APS/123-QED}
	
	\title{Device-independent secret key rates via a post-selected Bell inequality} 
	\author{Sarnava Datta}
	\email{Sarnava.Datta@hhu.de}
	\author{Hermann Kampermann}
	\author{Dagmar Bru\ss}
	\affiliation{ Institut f\"{u}r Theoretische Physik III\\
		Heinrich-Heine-Universit\"at D\"usseldorf}
	\date{\today}
	
	\begin{abstract}
		In device-independent quantum key distribution (DIQKD)   the security is not based on any assumptions about the intrinsic properties of the devices and the quantum signals, but on the violation of a Bell inequality. We introduce a  DIQKD scenario in which an optimal  Bell inequality is constructed from the performed measurement data, rather than fixing beforehand a specific Bell inequality. Our method can be 
		employed in a general way, for any number of measurement settings and any number of outcomes. We provide an implementable DIQKD protocol and perform finite-size security key analysis for collective attacks. We compare our approach with related procedures in the literature and analyze the robustness of our protocol. We also study the performance of our method in several Bell scenarios as well as for random measurement settings.    	 	 
	\end{abstract}
	\maketitle
	\section{Introduction}
	Data security concerns are prevalent in the modern world. One of the most prominent domains of quantum communication is quantum key distribution (QKD) which allows to distribute a secure key between two (or more) parties, namely Alice and Bob, where the security is only based on the laws of quantum mechanics. Since the inception of QKD \cite{bennett1984proceedings}, a variety of QKD protocols \cite{ekert1991quantum,bennett1992quantum,bruss1998optimal,renner2008security,lo2005decoy,gottesman2004security,shor2000simple,scarani2009security,ma2005practical,lo2014secure,tomamichel2012tight} has been introduced. However, the security of these device-dependent protocols needs complete characterization of the devices, sources, and/or the channel between the parties. In a realistic scenario, the device can be not completely characterized, or could even be prepared by a  malicious eavesdropper (Eve). Furthermore, hacking of existing implementations that exploits experimental imperfections was demonstrated \cite{lydersen2010hacking, gerhardt2011full, zhao2008quantum}. To overcome these drawbacks, device-independent (DI) QKD was introduced \cite{mayers1998quantum}, where the security does not require any assumptions about the inherent properties of the devices, or the dimension of the Hilbert space of the quantum signals. The security of DI protocols is based on the observation of a loophole-free Bell inequality violation \cite{barrett2005no,acin2007device, acin2006bell, pironio2009device, hanggi2010device, hanggi2010efficient, masanes2014full, arnon2018practical, arnon2019simple, masanes2011secure, acin2006bell,  ribeiro2018fully, holz2019genuine, vazirani2014fully} which guarantees the quantum nature of the observed data. The length of the secret key will depend on the estimated violation of the Bell inequality. 
	\\ 
	In this article, we introduce a DIQKD scenario in which the Bell inequality is not agreed upon beforehand, but will be constructed from the observed probability distribution of the measurement outcomes. We follow a two-step process: From the input-output probability distribution, we construct a Bell inequality that leads to the maximum Bell violation for that particular measurement setting of Alice and Bob. Then we use this optimized Bell inequality and the corresponding violation to bound the secret key rate. 
	- Note that in \cite{nieto2014using,bancal2014more} the authors introduced an alternative approach to bound the device-independent secret key rate via the measurement statistics. We will relate and compare our method with theirs in the Results section (Sec. \ref{Results}). In particular, we show that our procedure is advantageous in the non-asymptotic regime.\\
	This paper is organized as follows. We start in Sec. \ref{Framework} by briefly reviewing classical and quantum correlations. Then we explain how to obtain the optimal Bell inequality from the observed probability distribution.	We lay the framework to provide a confidence interval for the Bell expectation value in Sec. \ref{statfluc}. We provide an implementable DIQKD protocol in Sec. \ref{protocol} and calculate the finite-size secret key rate in Sec. \ref{BellVal_skr}. In Sec. \ref{Results}, we illustrate our method with several examples.  
	
	\section{General Framework}\label{Framework}
	In this section, we review the concept of the classical correlation polytope in Sec. \ref{Correlation set} and, based on this, we explain 
	in Sec. \ref{Designing BI} how to construct Bell inequalities that are maximally violated by the measurement data. 
	\subsection{Set of correlations}\label{Correlation set}
	Consider a set-up for two parties\footnote{Note that our method can be extended in a 
	straightforward way to $n$ parties.} (namely, Alice and Bob) connected by a quantum channel. The parties perform local measurements on a joint quantum state. Let us assume that Alice and Bob have $ m_a $ and $ m_b $ measurement settings, respectively. Alice's set of measurement settings is denoted as $ X=\{1,\cdots,m_a\} $, and Bob's set of measurement settings as $ Y=\{1,\cdots,m_b\} $. To estimate the probability distribution from the experimental data, we have to use the measurement device $ N $ times in succession. We assume that the devices behave independently and identically (i.i.d.) in each round, i.e. the results of the $ i^{th} $ round are independent of the past $ i-1 $ rounds. The setting of the $i^{th}$ round is denoted as $ x_i \in  X$ for Alice and $ y_i \in Y $ for Bob. Each of these measurement settings has $d$ outcomes which are denoted as $a_i \in A=\{1,\cdots, d\}$ for Alice and $b_i \in B=\{1,\cdots, d\}$ for Bob. We call this the $ [(m_a,m_b),d] $ scenario, i.e. 2 parties with $ (m_a,m_b) $ measurement settings and $ d $ outcomes each. When both parties have an equal number of measurement settings, i.e. $m_a = m_b = m$, we will denote this as $[m,d]$ scenario. The joint probability of getting outcome $a$ when Alice is using the measurement setting $x \in X$ and $b$ when Bob uses the measurement setting $y \in Y$ is denoted as $P(A^{a}_{x}B^{b}_{y})$. All these joint probabilities will be collected in a probability vector
		\begin{equation}\label{Obs. Prob. dist.}
		\textbf{P}:= [P(A_{x}^{a}B_{y}^{b})] \, ,
		\end{equation}
		where $x \in X=\{ 1, \cdots, m_a \}$, $y \in Y=\{ 1, \cdots, m_b \}$, $ a \in A=\{ 1,2, \cdots, d \}$ and $ b \in B=\{ 1,2, \cdots, d \}$. The associated probability space is of dimension 
		\begin{equation}\label{hyperplane dim}
		D_{m_a,m_b}^d := m_a m_b d^{2} \, .
		\end{equation} 
	\\
	The set of all probabilities that represent a classical or locally real theory forms a convex polytope \cite{pitowski1989quantum,fine1982hidden,pitowsky1991correlation}. We denote this polytope as $\mathcal{P}$. Any probability distribution which is not contained in $\mathcal{P}$ shows non-classical or quantum behaviour and can be witnessed by the violation of a Bell inequality \cite{bell1964einstein}. As illustrated in \cite{szangolies2017device}, the polytope of classical correlations can be characterized by its extremal points $\textbf{v}_p$, where $p=\{1, 2, \cdots, d^{m_a+m_b}\}$, and $\textbf{v}_p$ has entries from the set \{0,1\}. The extremal points of the polytope correspond to deterministic strategies. Every classical correlation $\textbf{P}_{cl} \in \mathcal{P}$ can be written as a convex combination of all the deterministic strategies as
	\begin{equation}\label{classical decompostion}
	\textbf{P}_{cl}=\sum_{p=1}^{d^{m_a+m_b}}\lambda_p \textbf{v}_p \, ,
	\end{equation}
	where $\lambda_p \geq 0$ and $\sum_{p=1}^{d^{m_a+m_b}}\lambda_p=1$. This subsequently implies that every observed probability distribution which cannot be decomposed as shown in Eq.~(\ref{classical decompostion}) violates at least one Bell Inequality. 
	\subsection{Designing Bell inequalities}\label{Designing BI}
	\begin{figure}[h]
		\includegraphics[height=6cm, width=7cm]{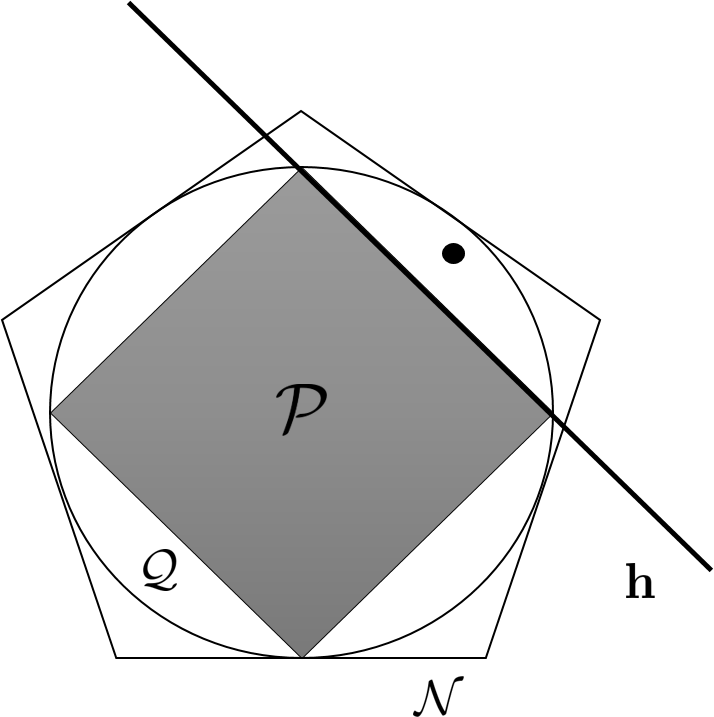}
		\caption{A sketch for the set of correlations. All
			classical probabilities form a convex polytope $\mathcal{P}$, which is embedded in the set $\mathcal{Q}$ of quantum correlations, which in turn is a subset of the non-signalling polytope $\mathcal{N}$. The Bell inequality is specified by the
			vector $ \textbf{h} $ defining a hyperplane which separates all vertices $ v_p $ from the observed probability distribution $ \textbf{P}$ (the black point situated outside the classical polytope $ \mathcal{P} $).}
		\label{Classical_Polytope}
	\end{figure}	
	Consider the $ [(m_a,m_b),d] $ scenario where the parties receive the measurement data $\textbf{P}$. In order to extract a secret key from these classical measurement data, they need to violate a Bell inequality. As shown in \cite{szangolies2017device}, this scenario can be translated to a linear separation problem. For illustration, see Fig. \ref{Classical_Polytope}. Bell inequalities correspond to hyperplanes in the probability space that separate the classical correlation polytope $\mathcal{P}$ from the set of all genuine quantum correlations $\mathcal{Q} \setminus \mathcal{P}$. Such hyperplanes are specified by a normal vector $\textbf{h} \in \mathbb{R}^{D^{d}_{m_a,m_b}}$, with the dimension given in Eq.~(\ref{hyperplane dim}). If $\textbf{P} \in \mathcal{Q} \setminus \mathcal{P}$, there exists at least one hyperplane $ \textbf{h} $ that separates all the vertices $ \textbf{v}_p $ of $ \mathcal{P} $ from the observed probability distribution $ \textbf{P} $. We set the objective of the linear program to find the hyperplane vector $ \textbf{h} $ corresponding to the Bell inequality which is maximally violated by the measurement data $ \textbf{P} $. This optimization problem can be formulated as:  
	\begin{equation}\label{linear_optimization}
	\begin{aligned}
	& \underset{\textbf{h},c}{\text{max}}
	& & \textbf{h}^T \textbf{P}-c\\
	& \text{subject to}
	& & \textbf{\textbf{h}}^T \textbf{v}_{p} \leq c \quad \forall \quad p \in \{1,\cdots,d^{m_a+m_b} \} \\
	&&& \textbf{h}^T \textbf{P} > c \\
	&&& -1 \leq h_i \leq 1 \quad \forall \quad i \in \{1,\cdots,D_{m_a,m_b}^d \}
	\end{aligned}
	\end{equation}
	with the classical bound c. The additional constraint imposed on the elements of $h_i$ of the hyperplane vector $ \textbf{h} $  keeps the maximization bounded. The chosen boundaries of $h_i$ do not influence the result of the optimization problem besides being a global scaling factor. The hyperplane found in this manner has the form
    \begin{equation}\label{hyperplane}
	\textbf{h}=[h_{A_x B_y}^{ab}] \, , 
	\end{equation}
	where here and in the following
	$x \in X=\{ 1, \cdots, m_a \}$, $y \in Y=\{ 1, \cdots, m_b \}$, $ a \in A=\{ 1,2, \cdots, d \}$ and $ b \in B=\{ 1,2, \cdots, d \}$. Thus, the Bell inequality found by the optimization and specified by the hyperplane vector \textbf{h} is given as
	\begin{equation}\label{Hyperplane BI}
	B[\textbf{P}]=\sum_{a ,b , x , y }
	h_{A_x B_y}^{ab}P(A_{x}^{a}B_{y}^{b}) \leq c \, .
	\end{equation}
	Eq.~(\ref{Hyperplane BI}) represents the Bell inequality that is maximally violated by the observed probability distribution $ \textbf{P} $. Note that, if $ \textbf{P} \in \mathcal{P}$, the optimization problem Eq.~(\ref{linear_optimization}) is infeasible and no Bell inequality can be found.
	
	\section{Statistical Fluctuations and their Estimation}\label{statfluc}
	So far, we have concentrated on the ideal asymptotic case, that is, using the exact probabilities as entries of the observed probability distribution $\textbf{P}$. However, in a real experiment, one does not have access to probabilities, but only to frequencies that are subject to statistical uncertainties and systematic errors. Since systematic errors mostly arise from specific experimental settings, we solely focus on the theoretical framework and concentrate on statistical fluctuations as they lead to uncertainties in the observed Bell violation. \\
	Let Alice and Bob perform N rounds of measurements. The number of instances when Alice chooses measurement $x \in X$ and Bob chooses measurement $y \in Y$ is denoted by $N_{x,y}$. In a real experiment, instead of having access to joint probabilities, we estimate them by the joint frequencies $\hat{P}(A_{x}^{a}B_{y}^{b})=\frac{N(a,b,x,y)}{N_{x,y}}$. Here $N(a,b,x,y)$ is the number of occurrences of the corresponding input-output pair.   \\
	
    The Bell value $B[\hat{\textbf{P}}]$ is a function of the joint frequencies,
	\begin{align}\label{estimatedBI}
	B[\hat{\textbf{P}}] &= h_{A_x B_y}^{ab}\hat{P}(A_{x}^{a}B_{y}^{b}) \, , \quad 
	\end{align}
	see also Eq.~(\ref{Hyperplane BI}).
	Let $\chi(e)$ be an indicator function for a particular event $e$, i.e. $\chi(e)=1$ if the event $ e $ is observed, $\chi(e)=0$ otherwise. We introduce a random variable 
	\begin{equation*}
	\hat{B}_i = \sum_{a ,b , x , y }
	h_{A_xB_y}^{ab}\frac{\chi(a_i=a,b_i=b,x_i=x,y_i=y)}{\hat{p}(x_i=x,y_i=y)} \, ,
	\end{equation*}
	where  $\hat{p}(x_i=x,y_i=y)=\frac{N_{x,y}}{N}$ is the input joint frequency distribution. We get $\frac{1}{N}\sum_{i=1}^{N}\hat{B}_{i}=B[\hat{\textbf{P}}]$. Defining
	\begin{align*}
	q_{\mathrm{min}} &= \min_{a ,b , x , y } \, 
	\frac{h_{A_xB_y}^{ab}}{\hat{p}(x_i=x,y_i=y)} , \\
	q_{\mathrm{max}} &=  \max_{a ,b , x , y } \, 
	\frac{h_{A_xB_y}^{ab}}{\hat{p}(x_i=x,y_i=y)} ,
	\end{align*}
	we have $ q_{\mathrm{min}} \leq \hat{B}_{i} \leq q_{\mathrm{max}} $.
	We define $\gamma := q_{\mathrm{max}}-q_{\mathrm{min}}$. By using Hoeffding's inequality \cite{hoeffding1963large,hoeffding1994probability} (see Lemma. \ref{Hoeffding lemma} in the appendix), we can bound the deviation $ \delta $ of the Bell value obtained by the frequencies from the asymptotic value by a probability:
	\begin{equation}\label{hoeffding ineq}
	\text{Pr}\bigg(B[\textbf{P}] \ge B[\hat{\textbf{P}}] - \delta \bigg) \ge 1-\epsilon ,
	\end{equation}
	with
	\begin{equation}\label{epsilon-delta relation}
	    \epsilon= \exp{\bigg(-\frac{2N\delta^2}{\gamma^2}\bigg)} \, .
	\end{equation}
	 For a given $\epsilon$ of a DIQKD protocol,  one can calculate the confidence interval $ \delta $ for the Bell value using Eq.~(\ref{epsilon-delta relation}).

	\section{DIQKD model and protocol}\label{protocol}
	Let us state the DIQKD protocol. We consider the i.i.d. scenario, where the devices will behave independently and identically in each round. The state distributed between the parties is also the same for each round of the protocol. Alice has $m$ measurement inputs $ x \in \{1,\cdots,m\} $. Each of the inputs has $d$ corresponding outputs $a \in \{ 1,\cdots,d \}$. Bob instead has $m+1$ measurement inputs $ y \in \{1,\cdots,m+1\} $. Each measurement input of Bob also has $d$ outputs $b \in \{1,\cdots,d\}$.
	\begin{enumerate}
		\item In every round of the protocol, the parties do the following:
		\begin{itemize}
			\item A state $\rho_{AB}$ is distributed between Alice and Bob.
			\item There are two types of measurement rounds, namely raw key generation rounds and parameter estimation rounds. According to a preshared random key $ T $, Alice and Bob choose a random $T_i=\{0,1\}$ such that $\mathrm{Pr}(T_i=1)=\xi$. If $T_i=0$, Alice and Bob choose the measurement input $(x=1,y=m+1)$ to generate the raw key. Otherwise, Alice and Bob choose the measurement inputs  $x \in \{1,\cdots,m\}$ and $y \in \{1,\cdots,m\}$, respectively, uniformly at random. These cases will be denoted as parameter estimation rounds. 
			\item The parties record their inputs and outputs as $(x_i, y_i)$ and $(a_i,b_i)$. After N rounds of measurement, we denote the input bit strings as $ X^N $ and $ Y^N $, and output bit strings as $ A^N $ and $ B^N $ for Alice and Bob, respectively.
		\end{itemize}
		
		\item Alice and Bob publicly reveal their measurement outcomes of the parameter estimation rounds. They divide the parameter estimation rounds' data into three sets. From the first set, Alice and Bob estimate the frequencies $\hat{\textbf{P}}_1 =[\hat{P}(A_{x}^{a}B_{y}^{b})]$ (see Eq.~(\ref{Obs. Prob. dist.})). If $\hat{\textbf{P}}_1$ is inside the classical correlation polytope $\mathcal{P}$, the protocol aborts. Otherwise, they construct an optimal Bell inequality by solving the linear optimization in Eq.~(\ref{linear_optimization}). Then Alice and Bob use the data from the second set to calculate the Bell value $B[\hat{\textbf{P}}_2]$. They then bound  the deviation of this estimated Bell value $B[\hat{\textbf{P}}_2]$ from the real Bell value $ B[\textbf{P}] $ by (see  Eq.~(\ref{hoeffding ineq}))
		\begin{equation}\label{BellVal deviation ineq}
		\text{Pr}(B[\textbf{P}] \ge B[\hat{\textbf{P}}_2] -\delta_{est}) \ge 1 - \epsilon_{est} \, ,
		\end{equation}
		where $\epsilon_{est}=\exp{\bigg(-\frac{2N\xi\delta_{est}^2}{3\gamma^2}\bigg)}$ and $\frac{N\xi}{3}$ are the number of measurement rounds used to estimate the Bell value $B[\hat{\textbf{P}}_2]$.\\
		The parties will use the Bell inequality and corresponding violation $B[\hat{\textbf{P}}_2]-\delta_{est}$ as a hypothesis in the experiment. From the data of the third set, the parties calculate the Bell value $B[\hat{\textbf{P}}_3]$. For an honest implementation, the protocol aborts if the Bell value $B[\hat{\textbf{P}}_3]$ is smaller than  $B[\hat{\textbf{P}}_2]-\delta_{est}$. 
		
		\item  Furthermore the parties need to estimate the QBER $ Q $ to bound the error correction information. Alice and Bob publicly reveal the measurement outcomes from $ N\eta $ randomly sampled key generation rounds to estimate the QBER. The QBER of the raw-key can be upper bounded with high probability using the tail inequality (see Lemma \ref{QBER_lemma} in the Appendix):
		\begin{equation}\label{QBER deviation ineq}
		\text{Pr}[Q \ge \hat{Q}+\gamma_{est}(N(1-\xi-\eta),N\eta,\hat{Q},\epsilon_{est}^{\gamma})] > \epsilon_{est}^{\gamma} \, ,
		\end{equation}
		where $ \gamma_{est}(N(1-\xi-\eta),N\eta,\hat{Q},\epsilon_{est}^{\gamma}) $ is the positive root of the following equation:
		\begin{align}
		\ln{ \binom{N(1-\xi-\eta)\hat{Q}+N(1-\xi-\eta)\gamma_{est}}{N(1-\xi-\eta)} } \\ \nonumber
		+ \ln{\binom{N\eta\hat{Q}}{N\eta}} \\ \nonumber
		= \ln{\binom{(N(1-\xi)\hat{Q}+N(1-\xi-\eta)\gamma_{est}}{N(1-\xi)}} + \\ \nonumber
		\ln{\epsilon_{est}^{\gamma}} \, .
		\end{align}
		Thus we can deduce that the QBER $ Q $ is not larger than $ \hat{Q} + \gamma_{est} $ (estimated QBER + statistical correction) with very high probability of $ 1-\epsilon_{est}^{\gamma} $.
		
		\item Alice and Bob use an one-way error correction (EC) protocol to obtain identical raw keys $K_A$ and $K_B$ from their bit strings $A^N$ and $B^N$. During the process of error correction, Alice communicates $O_{EC}$ to Bob such that he can guess the outcomes $A^N$ of Alice. If EC aborts, they abort the protocol. In an honest implementation, this happens with probability at most $\epsilon_{EC}^c$. Otherwise, they obtain error corrected raw keys $K_A$ and $K_B$ \cite{grasselli2020quantum,beaudry2015assumptions,tomamichel2012tight,scarani2008security}. The probability that Alice and Bob do not abort but hold different raw keys $K_A \neq K_B$ is at most $\epsilon_{EC}$. For details, see Appendix \ref{Estimation of leak}.\\
		
		When the real QBER $ Q $ is greater than $ \hat{Q} + \gamma_{est} $ (which happens with probability $ \epsilon_{est}^{\gamma} $), the hashed values of keys belonging to Alice and Bob (which is sent from Alice to Bob to check if the error correction successful, see Appendix \ref{Estimation of leak} for details) are different with high probability \cite{grasselli2020quantum}. This results in the abortion of the implemented error correction protocol. Thus, we can upper bound the error correction abortion probability $\epsilon_{EC}^c$ by $ \epsilon_{est}^{\gamma} $.
		
		\item Alice and Bob apply a privacy amplification protocol to obtain a secure final key $\Tilde{K}_A = \Tilde{K}_B$ of length $l$ that is close to be uniformly random and independent of the adversary’s knowledge.
	\end{enumerate}
	
	\section{Secret key rate}\label{BellVal_skr}
	To provide a lower bound on the device-independent secret key rate, one has to estimate two terms. One is the conditional von Neumann entropy $ H(A|X,E) $ and the other one is the error correction information $ H(A|B) $ of the raw key \cite{devetak2005distillation}. To estimate the latter, one can follow the footsteps of \cite{arnon2019simple,murta2019towards}, the detailed derivation is shown in Appendix \ref{Security key analysis}. For the estimation of the conditional von Neumann entropy $ H(A|X,E) $, we lower bound it by the conditional min-entropy $ H_{min}(A|X,E)=-\log_{2}P_g(A|X,E) $ (see Eq.~(\ref{von-Neumann min-entropy reln})) \cite{konig2009operational}, where $ P_g(A|X,E) $ is Eve's guessing probability about Alice's $ X $-measurement results conditioned on her side information $ E $. $ P_g(A|X,E) $ can be upper bounded by a function $ G_x $ of the estimated Bell violation $B[\textbf{P}]$ \cite{masanes2011secure} by solving a semi-definite programme \cite{johnston2016qetlab} i.e 
	\begin{equation} \label{Guessing prob bound}
	    P_{g}(A|X,E) \leq G_x(B[\textbf{P}]) \, .
	\end{equation}
	In real-life experiments, one does not have access to the probabilities. Instead, one has to deal with the frequencies. In Sec. \ref{protocol}, we discussed that the protocol will abort if the observed Bell violation $ B[\hat{\textbf{P}}_3] $ in the hypothesis testing is smaller than $B[\hat{\textbf{P}}_2]-\delta_{est}$. We need to take into account that the observed Bell violation $ B[\hat{\textbf{P}}_3] $ is calculated from a finite number of rounds. To infer the real Bell violation of the i.i.d. implementation, we make use of Hoeffding’s inequality to define a confidence interval $\delta_{con}$, and the associated error probability $\epsilon_{con}$. We bound the probability of wrongly accepting the hypothesis with the error probability $\epsilon_{con}$ by:  
	\begin{align}
	\text{Pr}(B[\hat{\textbf{P}}_2]-\delta_{est} \ge B[\hat{\textbf{P}}_3]+\delta_{con} ) < \epsilon_{con}  \nonumber \\
	\Rightarrow \text{Pr}(B[\hat{\textbf{P}}_2]-\delta_{est}-\delta_{con} \ge B[\hat{\textbf{P}}_3] ) < \epsilon_{con} \ .
	\end{align}
	Therefore given that Alice and Bob do not abort the protocol, we infer that the Bell violation of the system under consideration is higher than $B[\hat{\textbf{P}}_2]-\delta_{est}-\delta_{con}$ (with maximun $\epsilon_{con}$ probability of error). We consider the worst possible scenario and use the Bell violation $B[\hat{\textbf{P}}_2]-\delta_{est}-\delta_{con}$ to upper bound the guessing probability $P_{g}(A|X,E)$ via a semi-definite programme 
	\begin{equation}\label{SDP}
	\begin{aligned}[t]
	& \underset{ \rho,\{A(a|x)\}, \{B(b|y)\}}{\text{max:}}\quad P_{g}(A|X,E)\\
	& \text{subject to:} \quad
	\mathrm{Tr}(\rho \mathcal{G})=B[\hat{\textbf{P}}_2]-\delta_{est}-\delta_{con} \, .
	\end{aligned}  
	\end{equation}
	The guessing probability $P_{g}(A|X,E)$ is bounded by using the NPA-hierarchy \cite{navascues2007bounding, navascues2008convergent} up to level 2 in the optimization problem of Eq.~(\ref{SDP}). The optimization is performed using standard tools YALMIP \cite{Lofberg2004}, CVX \cite{cvx,boyd2004convex,gb08}, Ncpol2sdpa \cite{wittek2015algorithm} and QETLAB \cite{qetlab}. $\mathcal{G}$ is the Bell operator defined as
	\begin{align*}
	\mathcal{G} &= \sum_{a ,b , x , y }
	h_{A_xB_y}^{ab}A(a|x)B(b|y) \ .
	\end{align*}
	$A(a|x)$ and $B(b|y)$ are measurement operators for Alice and Bob, respectively, and  $\rho$ is the state shared between Alice and Bob. Hence the conditional von Neumann entropy $H_{\text{min}}(A|XYE,T=1)$ can be bounded by
	\begin{align} \label{min-entropy bound}
	H_{\text{min}}(A|XYE,T=1)_{\rho} & = -\log_{2}P(A|X,E) \\ \nonumber
	& \ge -\log_{2} G_{x} (B[\hat{\textbf{P}}_2]-\delta_{est}-\delta_{con})	\, 	.			
	\end{align}
	The function $G$ is defined in Eq.~(\ref{Guessing prob bound}). T=1 specifies that the outcomes of the parameter estimation rounds which are used for the estimation of the min-entropy. \\
    To bound the error correction information, we need to estimate the QBER $ Q $, i.e. the probability that Alice's and Bob's measurement outcomes in the key generation rounds differ. In Sec. \ref{protocol}, we have discussed that we can upper bound the QBER $ Q $ of the raw key with at least $1-\epsilon^{\gamma}_{est}$ probability by $ \hat{Q} + \gamma_{est} $. In Appendix \ref{Security key analysis}, we show that we can upper bound the von Neumann entropy $H(A \arrowvert B)$ \cite{pironio2009device,grasselli2020quantum}:
	\begin{equation}\label{QBER bound main text}
	H(A|B) \le f(\hat{Q} + \gamma_{est}) \, ,
	\end{equation}
	where $f(x)=h(x)+x\log_{2}(d-1)$. Here, $ d $ is the number of outcomes per measurement in the Bell scenario  \cite{bradler2016finite} and $ h $ is the binary entropy function.\\
	Using the bound on the min-entropy (see Eq.~(\ref{min-entropy bound})) and the QBER (see Eq.~(\ref{QBER bound main text})), we derive the finite-size secret key rate of a $\epsilon_{DIQKD}^s$-sound, $\epsilon_{DIQKD}^c$-complete (see Def. \ref{security defn} and  Appendix \ref{Security key analysis} for details) DIQKD protocol for collective attacks. The statement is as follows \cite{murta2019towards}: Either the protocol in Sec. \ref{protocol} aborts with probability higher than $1-(\epsilon_{con}+\epsilon^{c}_{EC})$ or an $(2\epsilon_{EC} + \epsilon_{s} + \epsilon_{PA})$)-correct-and-secret key of length
	\begin{align} \label{final SKR expression}
	l &\leqslant N( -\log_{2} G_x (B[\hat{\textbf{P}}_2]-\delta_{est}-\delta_{con}) - \\ \nonumber
	 &\quad (1-\xi-\eta)f(\hat{Q} + \gamma_{est})+(\xi+\eta) \log_{2}d )- \\ \nonumber
	 &\quad \sqrt{N}\bigg( 4\log_{2}\big(2\sqrt{2^{\log_{2} d}}+1 \big) \bigg(\sqrt{\log_{2}\frac{8}{{\epsilon'}^{2}_{EC}}}+\sqrt{\log_{2}\frac{2}{\epsilon_s^2}}\bigg)\bigg)  \\ \nonumber
	 &\quad -\log_{2} \bigg(\frac{8}{{\epsilon'}^{2}_{EC}} + \frac{2}{2-{\epsilon'}_{EC}}\bigg)-\log_{2} \frac{1}{\epsilon_{EC}}-2\log_{2}\frac{1}{2\epsilon_{PA}} \, ,
	\end{align}
	can be generated where $\epsilon_{DIQKD}^c \leq \epsilon_{est}+\epsilon^{\gamma}_{est}$ (for an honest implementation) and $\epsilon_{DIQKD}^s \leq 2\epsilon_{EC} + \epsilon_{s} + \epsilon_{PA}$. The expression in Eq.~(\ref{final SKR expression}) is derived in Appendix \ref{Security key analysis}. Table. \ref{DIQKD_protocol_parameter} lists all parameters of the DIQKD protocol. 
	\begin{table}
		\centering
		\begin{tabular}{|c| p{5cm} |} 
			\hline
			$ N $ & number of measurement rounds in the protocol \\
			\hline
			$ \xi $ & fraction of parameter estimation rounds for estimating the Bell violation \\
			\hline
			$ \eta $ & fraction of measurement rounds for estimating the QBER \\
			\hline
			$\epsilon_{s}$ & smoothing parameter \\
			\hline
			$ \epsilon_{EC} $, $ \epsilon'_{EC} $ & error probabilities of the error correction protocol \\ 
			\hline
			$ \epsilon_{EC}^{c} $ & probability of abortion of error correction protocol\\
			\hline
			$ \delta_{est} $ & width of the statistical interval for the Bell violation hypothesis test \\
			\hline
			$ \epsilon_{est} $ & error probability of the Bell violation hypothesis test \\
			\hline
			$ \delta_{con} $ & confidence interval for the Bell test \\
			\hline
			$ \epsilon_{con} $ & error probability of the Bell violation estimation \\
			\hline
			$ \gamma_{est} $ & width of the statistical interval for the QBER estimation \\
			\hline
			$ \epsilon^{\gamma}_{est} $ & error probability of the QBER estimation \\
			\hline
			$ \epsilon_{PA} $ & error probability of the privacy amplification protocol \\
			\hline	
			$ \epsilon_{DIQKD}^c $ & completeness parameter of the DIQKD protocol \\
			\hline	
			$ \epsilon_{DIQKD}^s $ & soundness parameter of the DIQKD protocol \\
			\hline	
		\end{tabular}
		\caption{ Parameters of the DIQKD protocol }
		\label{DIQKD_protocol_parameter}
	\end{table}	
	
	\section{Results}\label{Results}
	In this section, we illustrate the potential and the versatility of our method with examples. We choose $\epsilon_{DIQKD}^c=10^{-2} $, $\epsilon_{DIQKD}^s=10^{-5} $, $\epsilon_{EC}=10^{-10}$ as DIQKD parameters for all the examples showed in the following section.
	\subsection{ Scenario of $m$ measurements each, 2 outcomes }\label{Results 2m2} 
	We present the scenario with $m$ measurement settings for Alice and $m+1$ for Bob (where the outcomes of only $m$ measurement settings are used in the parameter estimation). Each of those measurement settings has 2 possible outcomes. Let the shared state between Alice and Bob be a maximally entangled Bell state $\arrowvert \psi \rangle = \frac{1}{\sqrt{2}}(\arrowvert 00 \rangle + \arrowvert 11 \rangle) $, mixed with white noise of probability $p$, i.e.
	\begin{equation}\label{BellState}
	\rho=(1-p)\arrowvert \psi \rangle \langle \psi \arrowvert+p\frac{\mathbb{1}}{4} \ ,
	\end{equation}
	with $ p \in [0,1] $. Both parties use $ \sigma_z $ as key generation measurements, resulting in the maximal possible correlation between the outcomes of Alice and Bob. \\ 
	In the case of $ m=2 $, consider the  measurement settings of Alice and Bob that maximally violate the CHSH inequality \cite{clauser1969proposed}, i.e.
	\begin{equation}\label{CHSH settings}
	\begin{aligned}
	x &=1 \Rightarrow \sigma_z \, , & x &=2 \Rightarrow \sigma_x \, , \\
	y&=1 \Rightarrow \frac{\sigma_z+\sigma_x}{\sqrt{2}} \, , & y &=2 \Rightarrow \frac{\sigma_z-\sigma_x}{\sqrt{2}} \, .
	\end{aligned}
	\end{equation}
	For the CHSH settings with different values of white noise $p$, we recover the stable hyperplane stated in Table. \ref{CHSH facet}. 
	
	\begin{table}[H]
		\centering
		\begin{tabular}{ c  c | c  c}
			1 & -1 & 1 & -1 \\
			-1 & 1 & -1 & 1 \\
			\hline
			1 & -1 & -1 & 1  \\
			-1 & 1 & 1 & -1 
		\end{tabular}
		\caption{Optimized Bell inequality for the measurement settings in Eq.~(\ref{CHSH settings}),
		performed on a Bell state. 
		Here, the entries of the hyperplane vector, see Eq.~(\ref{hyperplane}), are given in a tabular form. For their explicit ordering see Appendix  \ref{APP CG notation}.}
		\label{CHSH facet}
	\end{table}
	The secret key rate as a function of the number of measurement rounds for different values of white noise $ p $ is shown in Fig. \ref{CHSH_finitekey_plot}. The key rate generated by our method coincides with Ref.\cite{masanes2011secure} that uses a predetermined standard CHSH inequality.\\ 
	\begin{figure}[h]
		\includegraphics[height=7cm, width=9cm]{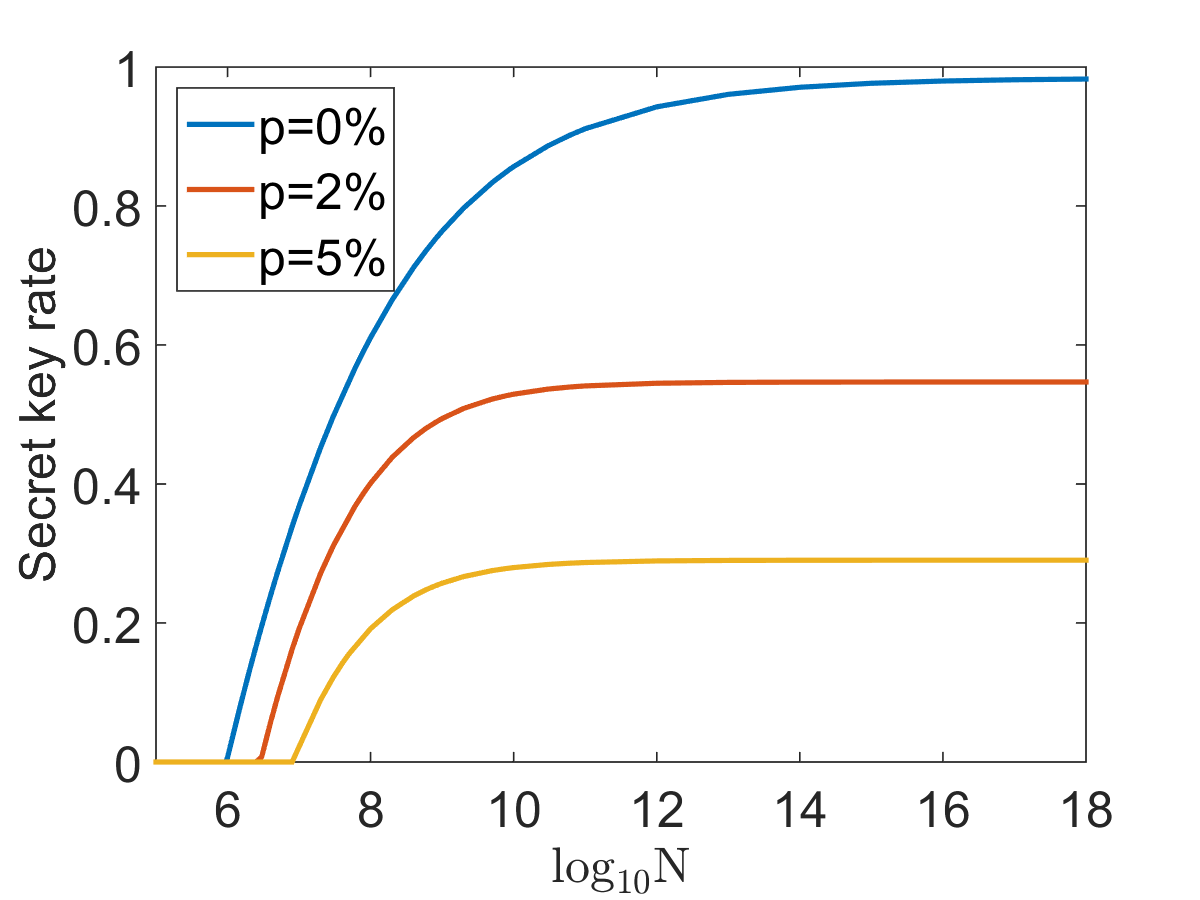}
		\caption{Secret key rate vs logarithm of the number of rounds N using the measurement settings of  Eq.~(\ref{CHSH settings}). The state shared between two parties is the noisy Bell state (defined in Eq.~(\ref{BellState})), where the noise is taken to be $p$ = 0.0 (blue), $p$ = 0.02 (red), $p$ = 0.05 (yellow). }
		\label{CHSH_finitekey_plot}
	\end{figure}
		
	In Ref.\cite{nieto2014using,bancal2014more}, the authors introduced an approach of bounding the device-independent secret key rate (DISKR) directly, by using the measurement data. In the asymptotic regime, this corresponds to using a Bell inequality that leads to the maximal DISKR for the precise setup. However, small changes in the parameters (e.g. imperfections on the measurement directions) or on the measured probability distribution may lead to different Bell inequalities corresponding to optimal secret key rate. We compare our method with Ref.\cite{nieto2014using,bancal2014more} in the finite key regime. We study two different Bell scenarios. For the $[2,2]$ scenario, we consider the CHSH settings (see Eq.~(\ref{CHSH settings})) and the noisy Bell state of Eq.~(\ref{BellState}) with $p=0$ (see graph (a) of Fig. \ref{2step_vs_fullprob_finitesize}) and $p=0.02$ (see graph (b) of Fig. \ref{2step_vs_fullprob_finitesize}). For the $[3,2]$ scenario (3 measurement settings each, 2 outcomes per measurement), we consider the setting: 
	\begin{equation}\label{CHAIN3 settings}
	\begin{aligned}
	x &=1 \Rightarrow \sigma_z \, , \\
	x &=2 \Rightarrow \sin{\frac{\pi}{3}}\sigma_x + \cos{\frac{\pi}{3}}\sigma_z \, , \\
	x &=3 \Rightarrow \sin{\frac{2\pi}{3}}\sigma_x + \cos{\frac{2\pi}{3}}\sigma_z \, , \\
	y &=1 \Rightarrow \sin{\frac{\pi}{6}}\sigma_x + \cos{\frac{\pi}{6}}\sigma_z \, , \\
	y &=2 \Rightarrow \sigma_x \, , \\ 
	y &=3 \Rightarrow \sin{\frac{5\pi}{6}}\sigma_x + \cos{\frac{5\pi}{6}}\sigma_z \, ,
	\end{aligned}
	\end{equation} 
	and use the noisy Bell state (Eq.~(\ref{BellState})) with $p=0, p=0.02$ (see graph (c) and (d) of Fig. \ref{2step_vs_fullprob_finitesize}).

	\begin{figure}[h]
		\centering
		\includegraphics[height=7cm, width=9cm]{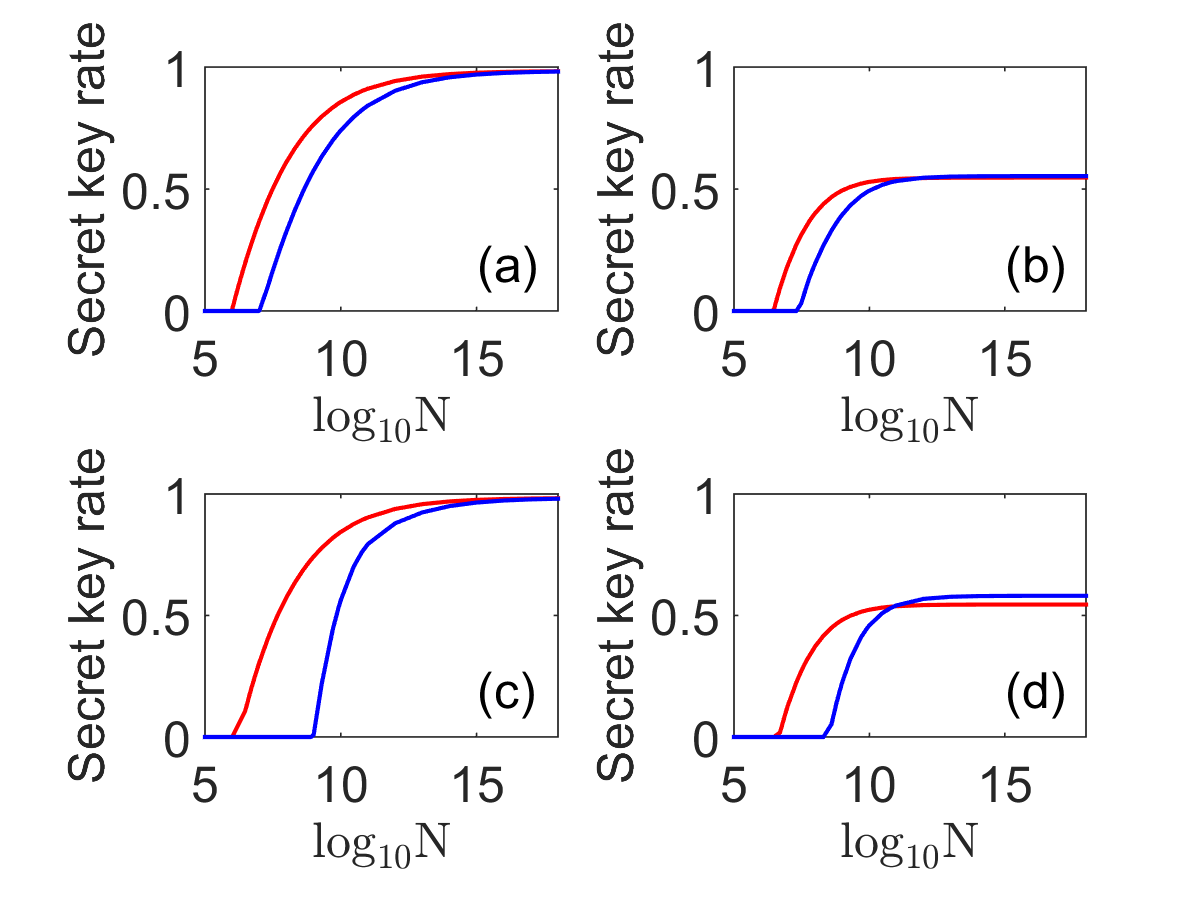}
		\caption{Achievable secret key rate as function of the number of measurement rounds N, comparing our method (red) and the method of Ref.~\cite{nieto2014using,bancal2014more} (blue), for
		a noisy Bell state with noise parameter $p$, see Eq.~(\ref{BellState}). Upper row:
		Measurement settings of  Eq.~(\ref{CHSH settings}), for (a)  $p=0$ and (b) $p=0.02$. 
		Lower row: Measurement settings of Eq.~(\ref{CHAIN3 settings}), 
		for (c)  $p=0$ and (d) $p=0.02$. }
		\label{2step_vs_fullprob_finitesize}
	\end{figure}
	To analyze the robustness, we incorporate fluctuations $\theta$ in the orientations in some measurement settings of Eq.~(\ref{CHAIN3 settings}) such that
	\begin{equation}\label{Tilted CHAIN3 settings}
	\begin{aligned}
	x &=1 \Rightarrow \sigma_z \, , \\
	x &=2 \Rightarrow \sin{(\frac{\pi}{3}-\theta)}\sigma_x + \cos{(\frac{\pi}{3}-\theta)}\sigma_z \, , \\
	x &=3 \Rightarrow \sin{(\frac{2\pi}{3}+\theta)}\sigma_x + \cos{(\frac{2\pi}{3}+\theta)}\sigma_z \, , \\
	y &=1 \Rightarrow \sin{(\frac{\pi}{6}+\theta)}\sigma_x + \cos{(\frac{\pi}{6}+\theta)}\sigma_z \, , \\
	y &=2 \Rightarrow \sigma_x \, , \\ 
	y &=3 \Rightarrow \sin{(\frac{5\pi}{6}-\theta)}\sigma_x + \cos{(\frac{5\pi}{6}-\theta)}\sigma_z \, .
	\end{aligned}
	\end{equation} 
	
	We use a noisy Bell state with $p=0.02$ (see Eq.~(\ref{BellState})) as the shared state between Alice and Bob. We use two approaches to compare the robustness of our method with Ref.\cite{nieto2014using,bancal2014more}. First, we set $\theta=\frac{\pi}{60}$ (see Eq.~(\ref{Tilted CHAIN3 settings})) and vary the number $N$ of measurement rounds (see (a) of Fig. \ref{2step_vs_fullprob_tiltedCHAIN3_finitesize}). Next we compare the methods for a range of deviations $\theta$ for $N=10^{10}$  measurement rounds (see (b) of Fig. \ref{2step_vs_fullprob_tiltedCHAIN3_finitesize}).
	
	We observe that the Bell inequality derived from our approach is stable against small fluctuations of the measurement directions or in the shared state. Our method can also generate a non-zero secret key by performing fewer measurement rounds in comparison with Ref.\cite{nieto2014using,bancal2014more} (see Fig. \ref{2step_vs_fullprob_finitesize} and Fig. \ref{2step_vs_fullprob_tiltedCHAIN3_finitesize}.
	This is because the effect of statistical corrections in the Bell inequality violation (see Eq.~(\ref{BellVal deviation ineq})) is smaller in our approach. These statistical corrections become insignificant for a high number of measurement rounds, such that the method of Ref.\cite{nieto2014using,bancal2014more} yields a higher secret key in the asymptotic regime.

	\begin{figure}%
    \centering
    \subfloat[\centering]{{\includegraphics[height=7cm,width=9cm]{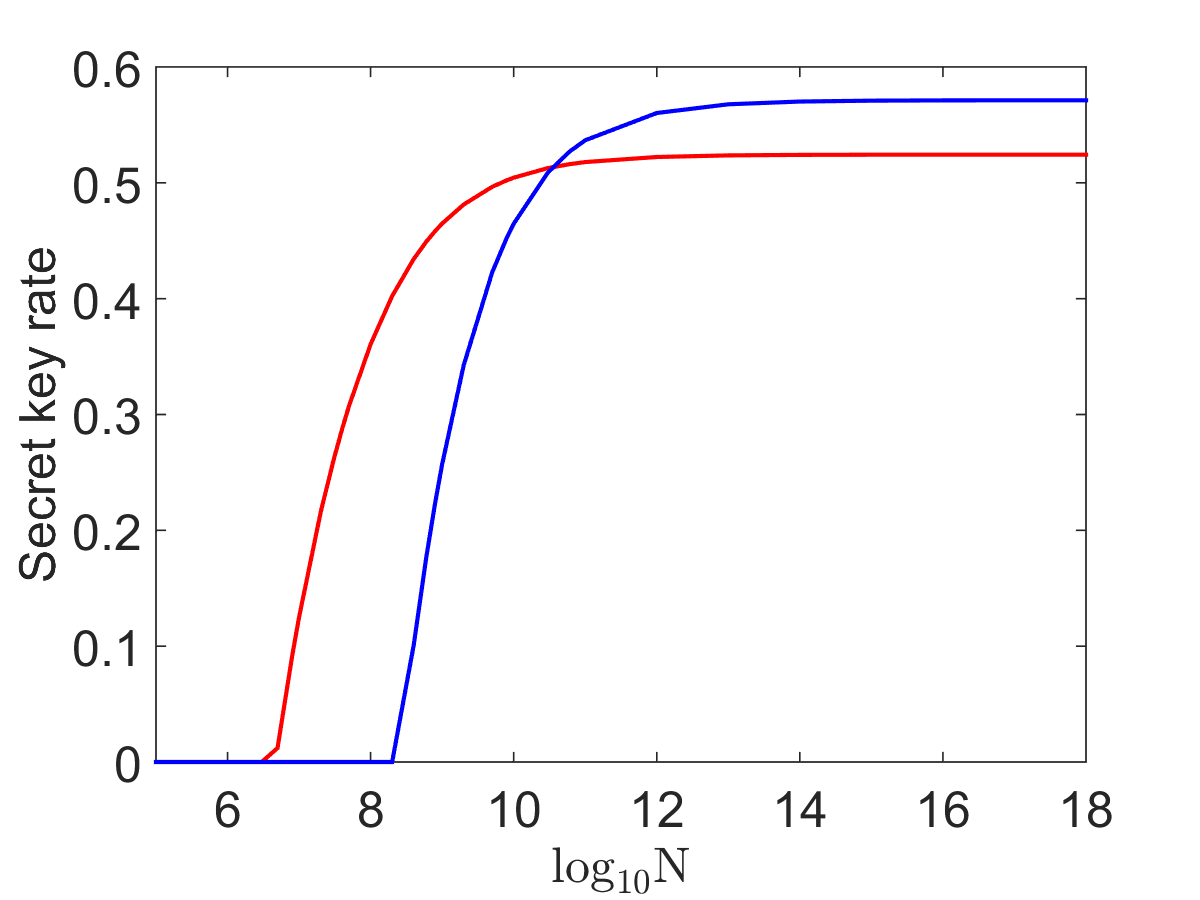} }}%
    \qquad
    \subfloat[\centering]{{\includegraphics[height=7cm, width=9cm]{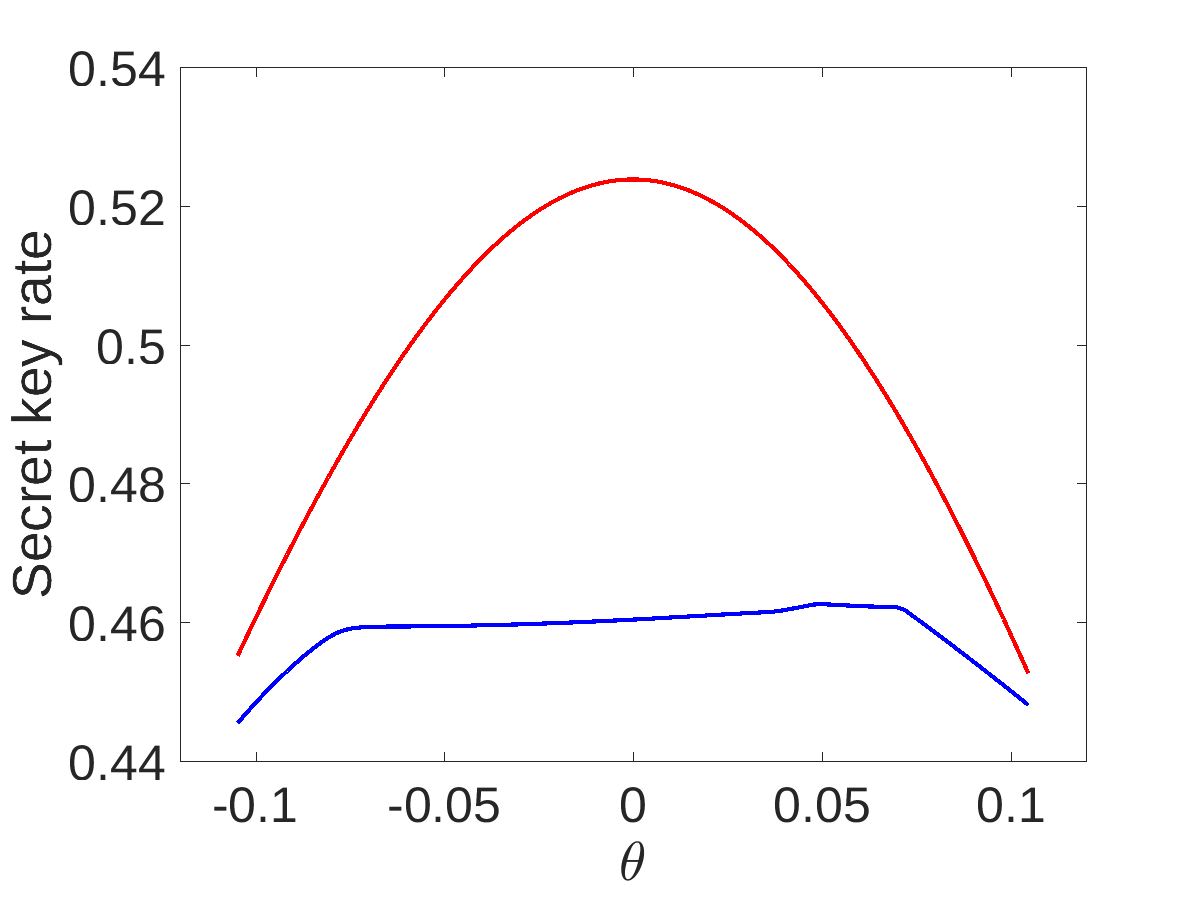} }}%
    \caption{Deviation of measurement direction: \\ (a):
	secret key rate vs logarithm of the number of
    rounds $N$ for our method (red) and the method of Ref.~\cite{nieto2014using,bancal2014more} (blue), with measurement settings of Eq.~(\ref{Tilted CHAIN3 settings}) where $\theta=\frac{\pi}{60}$, using a noisy Bell state with $p=0.02$ (see Eq.~(\ref{BellState})). (b): secret key rate vs deviation $\theta$ of the measurement settings in Eq.~(\ref{Tilted CHAIN3 settings}) for our method (red) and the method of Ref.~\cite{nieto2014using,bancal2014more} (blue), with $N=10^{10}$, using a noisy Bell state with $p=0.02$ (see Eq.~(\ref{BellState})).}
    \label{2step_vs_fullprob_tiltedCHAIN3_finitesize}%
    \end{figure}

	\begin{figure}[h]
		\includegraphics[height=7cm, width=9cm]{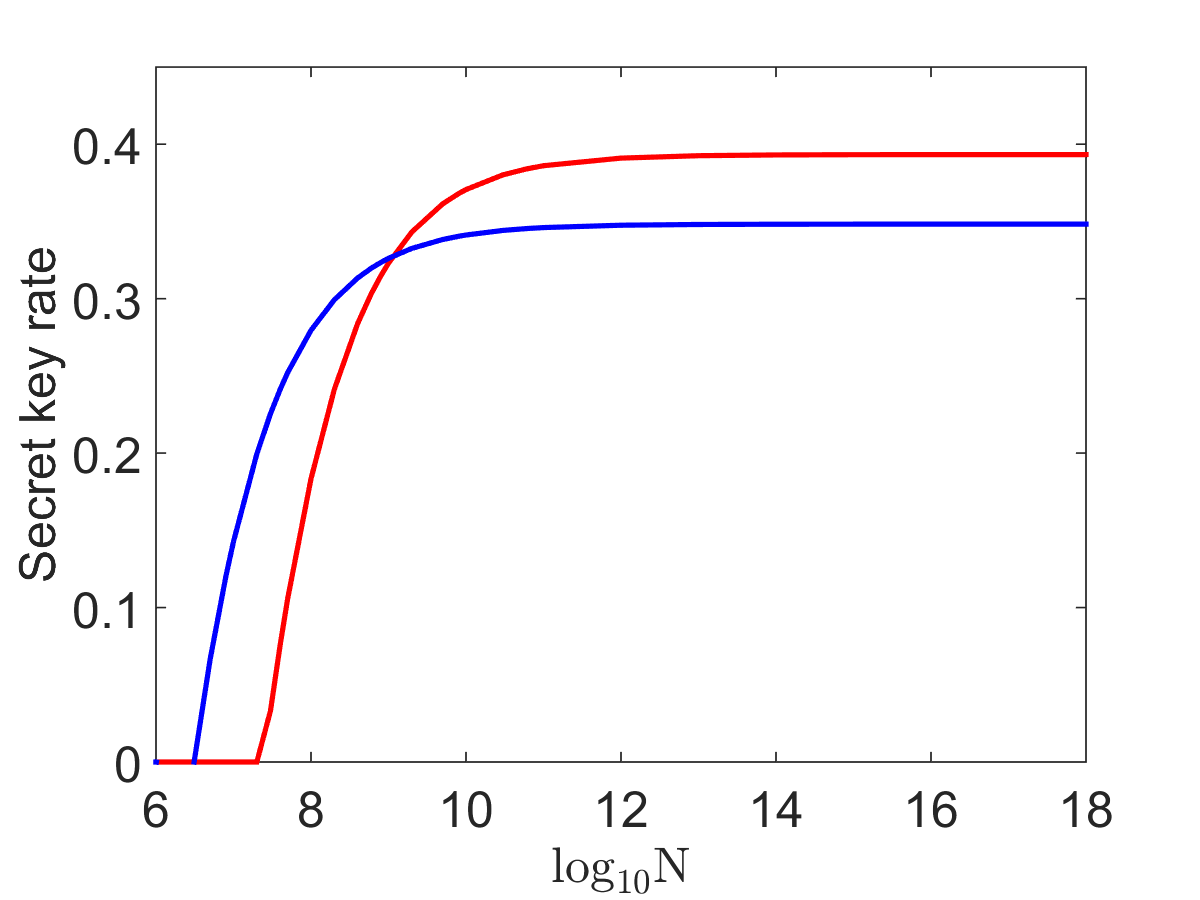}
		\caption{Improvement for more than 2 measurement settings: Secret key rate vs  $\log N$ for our method with measurement settings as given in Eq.~(\ref{3-meas settings_1 appendix}) (red), compared to the optimal subset of two measurement settings per party (blue).
		Here, the secret key rate for any subset of two measurement settings per party is calculated via the analytical expression in \cite{pironio2009device}, using the CHSH inequality. The shared state is a Bell state.}
		\label{optBIvschsh_plot}
	\end{figure} 
	We point out that our method can also have advantages w.r.t. the CHSH scenario, when the DI secret key rate is calculated via the analytical expression from Ref.~\cite{pironio2009device}: if non-optimal measurement settings were used, we can increase the key rate by employing additional measurement settings. As an example, we consider the  observed probability distribution originating from the maximally entangled Bell state  and the set of measurement settings listed explicitly in the Appendix, see Eq.~(\ref{3-meas settings_1 appendix}). With our method we can generate a higher secret key rate (for certain $N$) than using any subset of two measurement settings per party (and the analytical expression of \cite{pironio2009device}). See Fig. \ref{optBIvschsh_plot} for an illustration.
	
	If the probability distribution obtained by two non-optimal measurement settings per party does not lead to a non-zero secret key, adding another measurement setting per party and employing our strategy can be advantageous: For example, with non-optimal measurement settings in Eq.~(\ref{3-meas settings_2 appendix}) and the maximally entangled Bell state, one cannot extract a secret key, using our method or blindly using the CHSH inequality. By adding another set of measurements for Alice and Bob, as shown in Eq.~(\ref{add measurement appendix}), our method leads to a non-zero secret key rate.
	
	\subsection{Scenario of 2 measurements each, $d$ outcomes} \label{22d scenario}
	In this subsection, we analyse the scenario where each party has 2 measurement settings in the parameter estimation rounds (Bob has an additional measurement setting which will be used in key generation rounds), and each measurement has $d$ outcomes. 
	 The state shared between Alice and Bob is a maximally entangled state of two qudits, 
	 i.e. $\arrowvert \psi \rangle=\sum_{i=0}^{d-1}\frac{1}{\sqrt{d}}\arrowvert ii \rangle$, which is affected by white noise with probability $p$, i.e.
	\begin{equation}\label{maxentquditstate}
	\rho=(1-p)\arrowvert \psi\rangle \langle \psi \arrowvert+p\frac{\mathbb{1}}{d^2} \, .
	\end{equation}
	We consider the measurement settings from Ref.~\cite{acin2002quantum, collins2002bell}. 
	The measurement is carried out in 3 steps. In the first step Alice applies a unitary operation on her subsystem with only non-zero terms in the diagonal equal to $e^{\iota\Vec{\phi_{x}}(j)}$, where $x$ denotes Alice's measurement direction, i.e. $x \in \{1,2\}$, and $j=0,1,2,\cdots,d-1$. Similarly Bob applies a unitary operation on his subsystem with only non-zero terms in the diagonal equal to $e^{\iota\Vec{\varphi_{y}}(j)}$, where $y$ denotes Bob's measurement direction, i.e. $y \in \{1,2,3\}$.  These unitary operations are denoted by $U(\Vec{\phi}_{x})$ and $U(\Vec{\varphi}_{y})$ for Alice and Bob, respectively, where
	\begin{align*}
	\Vec{\phi}_{x}\equiv[\phi_{x}(0),\phi_{x}(1),\phi_{x}(2),\cdots,\phi_{x}(d-1)] \, ,  \\
	\Vec{\varphi}_{y}\equiv[\varphi_{y}(0),\varphi_{y}(1),\varphi_{y}(2),\cdots,\varphi_{y}(d-1)] \, .
	\end{align*}
	The values of these phases are chosen as
	\begin{equation}\label{CGLMP settings}
	\begin{aligned}
	\phi_1(j) &= 0 \, ,  & \phi_2(j) &= \frac{\pi}{d}j \, , \\
	\varphi_1(j) &= \frac{\pi}{2d}j \, , & \varphi_2(j) &= -\frac{\pi}{2d}j \, , & \varphi_3(j) &= 0 \, , 
	\end{aligned}
	\end{equation}
	with $j=0,1,2,\cdots,d-1$. We use $\{x=1,y=3\}$ for the key generation rounds and $\{x \in (1,2), \quad  y \in (1,2)\}$ for the parameter estimation rounds. The second step consists of Alice carrying out a discrete Fourier transform $U_{FT}$ and Bob applying $U^{*}_{FT}$. The matrix elements of the Fourier transform are defined as $(U_{FT})_{jk}$=$\exp{[(j-1)(k-1)2\pi \iota/d]}$, $(U^{*}_{FT})_{jk}$=$\exp{[-(j-1)(k-1)2\pi \iota/d]}$. Thus the concatenated unitaries for Alice and Bob are
	$V(\Vec{\phi}_{x}) \equiv U_{FT} \, U(\Vec{\phi}_{x})$ and
	$V(\Vec{\varphi}_{y}) \equiv U^{*}_{FT} \, U(\Vec{\phi}_{y})$,  
	respectively. 
	
	Finally, Alice and Bob carry out measurements in the computational basis $\arrowvert i  \rangle$.
	For $ d=3 $, we find via linear optimization, see Eq.~(\ref{linear_optimization}), the optimized Bell inequality as shown in Table~\ref{BI_CGLMP3}. The details of this representation of the Bell inequality are explained in Table~\ref{CGLMP hyperplane table [2,3]} of Appendix  \ref{APP CG notation}.
	\begin{table}[H]
		\centering
		\begin{tabular}{c c c | c c c}
			1 & -1 & 0 & -1 & 1 & 0 \\
			0 & 1 & -1 & 0 & -1 & 1 \\
			-1 & 0 & 1 & 1 & 0 & -1 \\
			\hline
			1 & 0 & -1 & 1 & -1 & 0  \\
			-1 & 1 & 0 & 0 & 1 & -1 \\
			0 & -1 & 1 & -1 & 0 & 1
		\end{tabular}
		\caption{Optimized Bell inequality 
		for the measurement described in the text, performed on a maximally entangled state of two qutrits. Here, the entries of the hyperplane vector, see Eq.~(\ref{hyperplane}), are given in a tabular form. For their explicit ordering see Appendix  \ref{APP CG notation}.
		}.
		\label{BI_CGLMP3}
	\end{table}
	The hyperplane in Table \ref{BI_CGLMP3} is equivalent to the CGLMP inequality \cite{collins2002bell, collins2004relevant}. If the parties share the non-maximally entangled state
	\begin{equation}\label{nonmaxentquditstate3}
	\arrowvert \Phi \rangle \equiv \frac{\arrowvert 00 \rangle+ 0.7923\arrowvert 11 \rangle+\arrowvert 22 \rangle}{\sqrt{2+0.7923^2}} \, ,
	\end{equation} 
	the CGLMP inequality is maximally violated, thus resulting in a significantly higher secret key rate, as  shown in Fig.~\ref{CGLMP3_ent_nonent_finitesize}. 
	This trend of generating a higher secret key rate using non-maximally entangled states is also observed for higher dimensions (i.e. $ d > 3 $). 
	\begin{figure}[htb]
		\includegraphics[height=7cm, width=9cm]{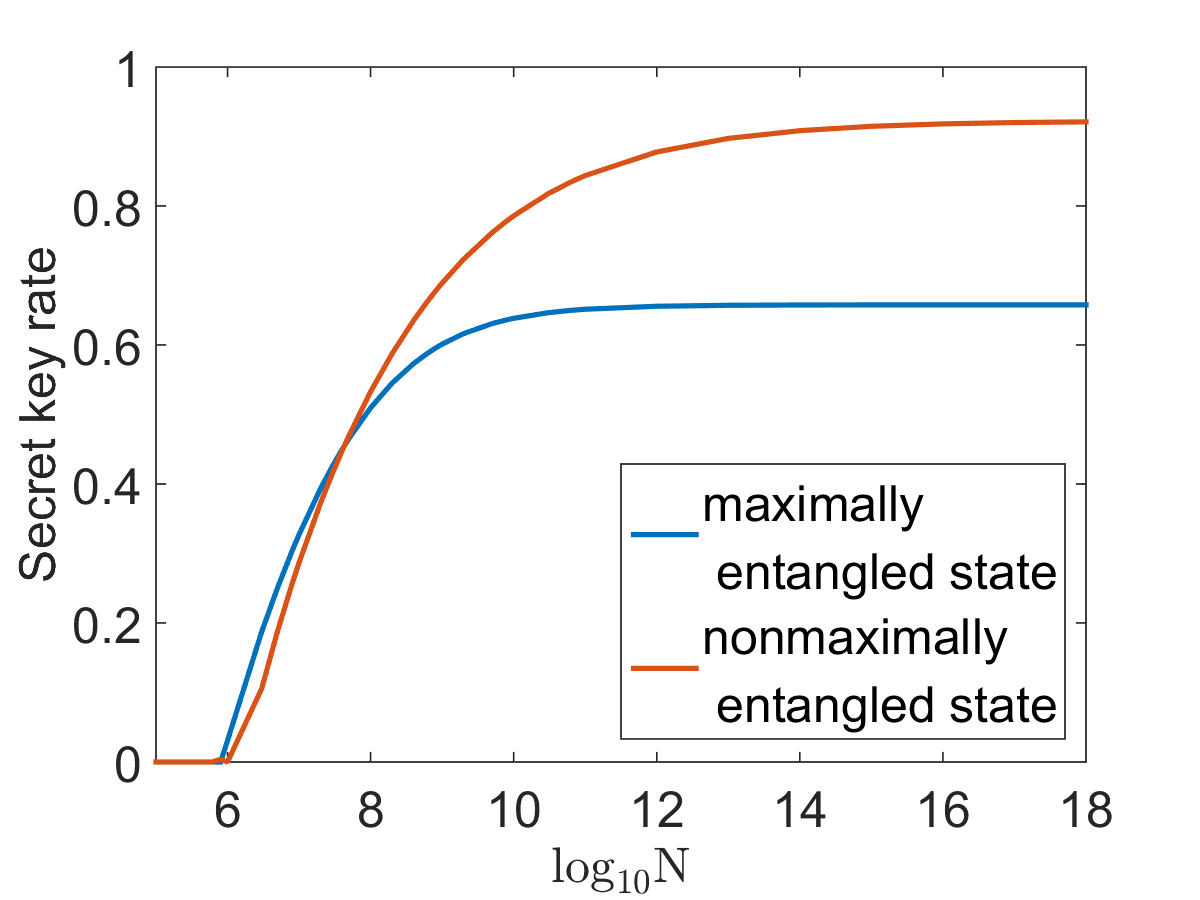}
		\caption{Secret key rate vs  $\log  N $, when performing the measurement described around Eq.~(\ref{CGLMP settings})
		on a maximally entangled state of two qutrits
		(blue) and on the non-maximally entangled state given in Eq.~(\ref{nonmaxentquditstate3}) (red).}
		\label{CGLMP3_ent_nonent_finitesize}	
	\end{figure} \\
	Note that in this scenario with $ d $ outcomes
	the maximum secret key rate is $\log_{2}d$. For a fair comparison, we have normalized the min-entropy (i.e. $-\log_{2}$ of the solution of the optimization problem of Eq.~(\ref{SDP})) by division with $\log_{2}d$ to get a rate per qubit dimension.
	\begin{figure}[htb]
		\includegraphics[height=7cm, width=9cm]{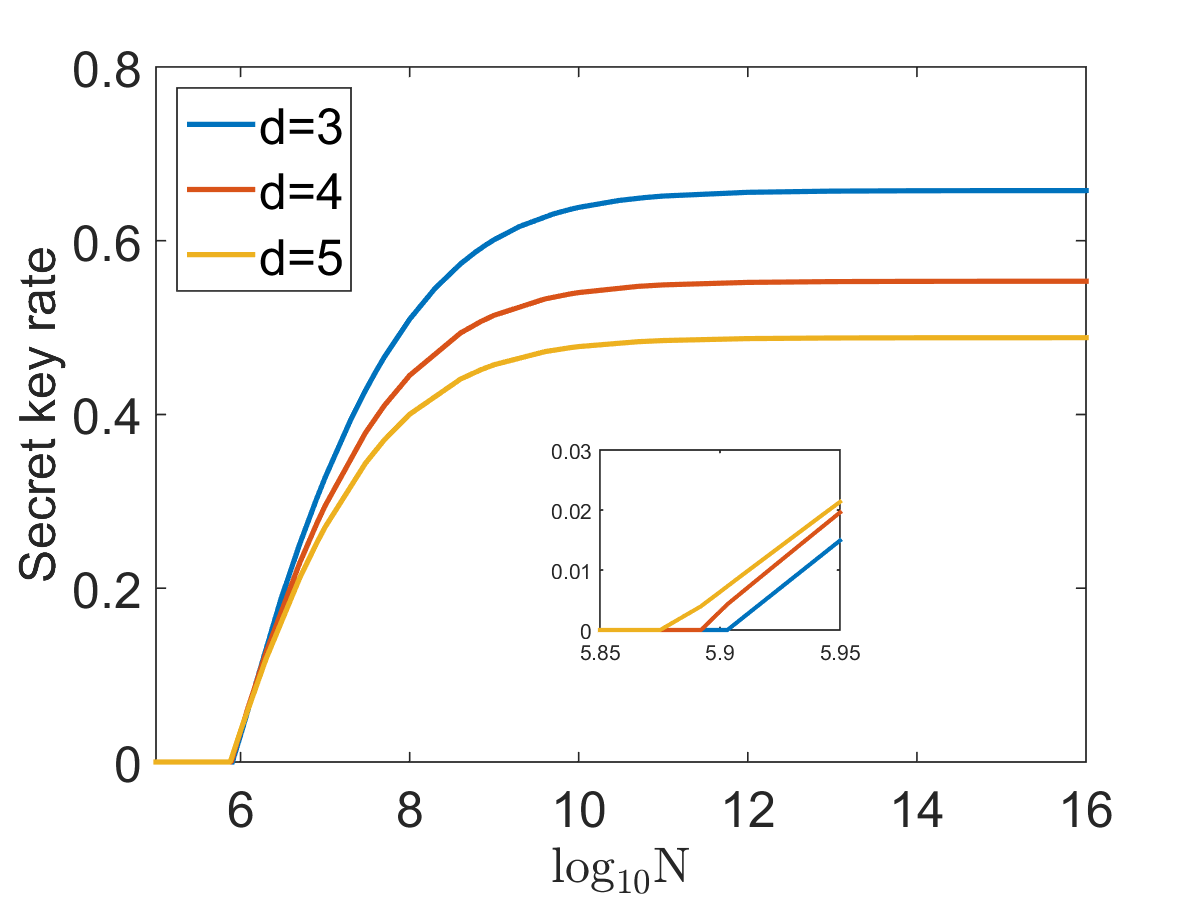}
		\caption{
		Secret key rate vs  $\log  N $, when performing the measurement described around Eq.~(\ref{CGLMP settings}) for $ d =3, 4, 5 $
		on a maximally entangled state of two $d$-dimensional subsystems.
		The inset graph shows a zoomed version in the region of low number of measurement rounds, demonstrating the cross-over of the curves.}
		\label{CGLMP_finitekey_comparison}	
	\end{figure}
	
	Comparing the DIQKD protocol with measurement settings as described around Eq.~(\ref{CGLMP settings}) for different $ d $ and the corresponding $ d $-dimensional maximally entangled state, see Eq.~(\ref{maxentquditstate}), the minimum number of measurement rounds required to have a non-zero secret key rate decreases slightly with increasing $ d $, see  Fig. \ref{CGLMP_finitekey_comparison}. This follows from the fact that the minimum number of measurement rounds required to have a non-zero Bell violation decreases with increasing $ d $. On the other side, the secret key is decreasing with increasing $ d $ (see Fig. \ref{CGLMP_finitekey_comparison}) when the number of measurement rounds is sufficiently high. The nonlocality of the resultant correlation is decreasing with increasing $ d $, which in turn results in the lower secret key.
	
	\subsection{Random measurement settings}
	 In this subsection, we analyse the case when Alice's and Bob's devices perform random measurements. We specifically focus on the fraction of events that leads to a non-zero secret key rate. First consider the $[m,2]$ scenario, i.e. $m$ measurement each, with 2 outcomes. The state shared between the parties is the noisy  Bell state as in  Eq.~(\ref{BellState}). We choose the raw key generation measurement operators $\{ x=1 \Rightarrow \sigma_z,  y=m+1 \Rightarrow \sigma_z \}$, in order to achieve correlated outcomes in the key measurement rounds and consequently have to exchange less error correction information. The remaining measurement operators are chosen randomly. Alice and Bob perform general unitary operators
	\begin{equation}
	U(\phi,\psi,\chi) = \begin{bmatrix}
	e^{\iota \psi}\cos{\phi} & e^{\iota \chi}\sin{\phi}\\
	-e^{-\iota \chi}\sin{\phi} & e^{-\iota \psi}\cos{\phi}
	\end{bmatrix}
	\end{equation}
	with parameters $ \psi,\chi \in [0,2\pi] $ and $ \phi \in [0,\frac{\pi}{2}] $ and then measure in the computational basis $ \{|0 \rangle, |1 \rangle \} $. This strategy is equivalent to choosing a random measurement.
	In Table \ref{Random measurement table}, we show the fraction of events that leads to a non-zero secret key rate  with random measurements. The statistics are based on $ 10^5 $ realizations. For the $[2,2]$ scenario, the optimization in Eq.~(\ref{linear_optimization}) will always lead to the CHSH inequality. Adding another measurement setting per party (i.e. the [3,2] scenario) significantly increases the probability of finding a hyperplane that produces a non-zero secret key rate. The first explanation of this fact is statistical. By increasing the number of settings, we increase the probability that some of them violate a Bell inequality even involving only two settings per party. Apart from that, the optimization in Eq.~(\ref{linear_optimization}) also provides some hyperplanes for the $[3,2]$ scenario that are independent of the hyperplanes for the $[2,2]$ scenario.
	
	\begin{table}[H]
		\centering
		\begin{tabular}{| c | c | c |}
			\hline
			& ($ m_a $,$ m_b $)=2 & ($ m_a $,$ m_b $)=3 \\
			\hline \hline
			p=0 \% & $\sim$ 28.6\% & $\sim$ 53.4\%  \\
			\hline
			p=1 \% & $\sim$ 18.3\% & $\sim$ 46.5\% \\
			\hline
			p=2 \% & $\sim$ 10.8\% & $\sim$ 36.8\%\\
			\hline
			p=3 \% & $\sim$ 6.4\% & $\sim$ 28.2\%\\
			\hline
			p=4 \% & $\sim$ 3.9\% & $\sim$ 18.5\%\\
			\hline
			p=5 \% & $\sim$ 2.2\% & $\sim$ 11.3\%\\
			\hline
		\end{tabular}
		\caption{Approximate probability of achieving a non-zero secret key rate in the $[m,2]$ scenario, for different white noise levels $p$ in the noisy  Bell state (see Eq.~(\ref{BellState})). The statistics are taken over $10^5$ realizations. Measurement settings of key generation rounds are fixed to be $\sigma_z$ for Alice and Bob. The remaining measurement settings are performed in random orientation. For each realization, $ 10^{12} $ measurement rounds are used to compute the finite key. }
		\label{Random measurement table}
	\end{table}
	 From the higher chance of Bell inequality violation, we obtain a higher chance of achieving a non-zero key. This result also reverberates the results of the nonlocal volume\footnote{The nonlocal volume is a statistical measure of nonlocality introduced in \cite{fonseca2015measure}. It is defined as the probability that the correlations, generated from randomly chosen projective measurements made on a given state $\arrowvert \psi \rangle$, violate any Bell inequality (a witness of nonlocality) by any extent. Generally, the nonlocal volume for a given state $\arrowvert \psi \rangle$ is obtained by $\int d\Omega f(\arrowvert \psi \rangle,\Omega)$, where one integrates over the measurement parameters $\Omega$ \cite{lipinska2018towards}. $f(\arrowvert \psi \rangle,\Omega)$ is an indicator function that takes the value 1 if the resultant correlations, generated from the state and measurements, are nonlocal. Otherwise, it will take the value 0.} in \cite{lipinska2018towards, de2017multipartite, de2020strength,fonseca2015measure,fonseca2018survey,barasinski2018volume}, which increases for the pure bipartite entangled state when more measurement settings for each party are used. We observe the same phenomenon in our case, regarding the secret key rate. As the nonlocal volume shrinks by adding noise, it also reduces the probability of producing a non-zero key rate.  \\
	
	\begin{table}[H]
		\centering
		\begin{tabular}{| c | c | c |}
			\hline
			& $ d=3 $ & $ d=4 $ \\
			\hline \hline
			p=0 \% & $\sim$ 6.4\% & $\sim$ 2.5\%  \\
			\hline
			p=1 \% & $\sim$ 2.2\% & $\sim$ 0\% \\
			\hline
			p=2 \% & $\sim$ 0.3\% & $\sim$ 0\%\\
			\hline
		\end{tabular}
		\caption{
		Approximate probability of achieving a non-zero secret key rate in the $[2,d]$ scenario, for  white noise of different probability $p$  added to the maximally entangled state 
		of two qudits (see Eq.~(\ref{maxentquditstate})). 
		All other details are as in Table~\ref{Random measurement table}.
		}
		\label{Random measurement table 22d}
	\end{table} 
	
	Let us now analyse the $[2,d]$ scenario (i.e. $d$ outcomes per measurement) with random measurement settings. The shared state is a noisy maximally entangled state of two qudits (see Eq.~(\ref{maxentquditstate})). We compute the approximate probability for achieving a non-zero secret key rate (see Table \ref{Random measurement table 22d}). The statistics are based on $10^5$ realizations. The measurements for key generation are in the computational basis. The remaining measurement settings are chosen randomly.
	
	We observe that for $ d \geq 3 $, the probability to extract a non-zero secret key is smaller compared to the case with only two outcomes. This follows from the fact that the non-local volume shrinks by increasing the dimension of the maximally entangled state. This results in a smaller probability of generating non-local correlations, and therefore a smaller chance of a Bell inequality violation \cite{fonseca2018survey} and smaller probability of a non-zero secret key.

	\section{Conclusions}
	Several protocols for device-independent quantum key distribution (DIQKD) have the common feature that they rely on the violation of a predetermined Bell inequality. We propose a robust DIQKD procedure where a suitable Bell inequality is instead constructed from the measurement data. This constructed Bell inequality leads to the maximum Bell violation for the particular set-up. Then we use the Bell inequality and its corresponding violation to bound the secret key rate via lower bounding the min-entropy. 
	\\
	 We provide a finite-size key analysis of our proposed procedure. We bound the statistical fluctuations of the Bell inequality violation by Hoeffding's inequality. However, we do not claim that our choice of concentration inequality \cite{massart2007concentration,boucheron2013concentration,chung2006concentration} is optimal for a finite number of measurement rounds. 
	- Note that our method could also be implemented for the estimation of global randomness in a device-independent randomness generation protocol.\\
	We have illustrated our method with several examples for different numbers of measurement settings and different numbers of outcomes, showing cases when our method yields a higher secret key rate than using the standard CHSH inequality. In comparison to related approaches (Ref.~\cite{nieto2014using,bancal2014more}), we provide examples where our approach needs  fewer number of measurement rounds to generate a non-zero secret key. We further showed the performance of our method in the case of random measurement settings. Finally, future work should address the use of more sophisticated methods of bounding the conditional von Neumann entropy \cite{schwonnek2020robust,brown2020computing}, which could increase the secret key rate, in comparison to the bounds based on the min-entropy.
		
	\section{Acknowledgements}
	The authors acknowledge support from the Federal Ministry of Education and Research (BMBF, Projects Q.Link.X and HQS). We also acknowledge support by the QuantERA project QuICHE, via the German Ministry for Education and Research (BMBF Grant No. 16KIS1119K). We thank Gláucia Murta, Federico Grasselli and Lucas Tendick for helpful discussions.
	
	\onecolumngrid
	\appendix
	\section{Definitions}
	We start with the definition of some quantities that will help us to derive the key rates for the DIQKD protocol.
	\begin{defn}[Min and max-entropy \cite{tomamichel2015quantum,renner2005simple}]\label{min-max entropy} 
		Let $\rho_{AB} \in \mathcal{P}(\mathcal{H}_A \otimes \mathcal{H}_B) $ and $ \sigma_{B} \in \mathcal{P}(\mathcal{H}_B) $. $  \mathcal{P}(\mathcal{H}_B) $ is the set of positive-semidefinite operators on the Hilbert space $\mathcal{H}_B$. 
		The min-entropy of $ \rho_{AB} $ conditioned on $ \sigma_{B} $ is 
		\begin{equation} \label{min-entropy Eq}
		H_{min}(\rho_{AB} \arrowvert \sigma_{B}) := -\log \lambda \, ,
		\end{equation}
		where $\lambda$ is the minimum real number such that $\lambda.(\mathbb{I}\otimes \sigma_{B})-\rho_{AB} \ge 0 $. The max-entropy of $ \rho_{AB} $ conditioned on $ \sigma_{B} $ is
		\begin{equation} \label{max-entropy Eq}
		H_{max}(\rho_{AB} \arrowvert \sigma_{B}) := \log \mathrm{Tr}((\mathbb{I} \otimes \sigma_{B})\rho_{AB}^0)\, ,
		\end{equation}
		where $\rho_{AB}^0$ denotes the projector onto the support of $\rho_{AB}$. 
	\end{defn}
	
	\begin{defn}[Smoothed min and max-entropy \cite{tomamichel2015quantum,vitanov2013chain}]\label{Smoothed entropy}
		For a quantum state $\rho_{AB}$ and $\epsilon \ge 0$, the smooth min-entropy of system $ A $ conditioned on $ B $ is defined as  
		\begin{equation} \label{smooth min-entropy Eq}
		H^{\epsilon}_{min}(A \arrowvert B) := \max_{\tilde{\rho}_{AB} \in \mathcal{B}^{\epsilon}(\rho_{AB})} H_{min}(A \arrowvert B)_{\tilde{\rho}_{AB}} \, ,
		\end{equation}
		and, the smooth max-entropy of system $ A $ conditioned on $ B $ is defined as   
		\begin{equation} \label{smooth max-entropy Eq}
		H^{\epsilon}_{max}(A \arrowvert B) := \min_{\tilde{\rho}_{AB} \in \mathcal{B}^{\epsilon}(\rho_{AB})} H_{max}(A \arrowvert B)_{\tilde{\rho}_{AB}} \, .
		\end{equation}
		$ \mathcal{B}^{\epsilon} $ is an $\epsilon$-ball of sub-normalized operators around the state $ \rho_{AB} $ defined in terms of the purified distance.
	\end{defn}
	

	Now we focus on the security parameters of quantum key distribution. The security of quantum key distribution can be split into two conditions. 
	\begin{defn}[Correctness \cite{renner2008security,arnon2019simple,murta2019towards}]\label{Correctness}
		A DIQKD protocol is $\epsilon_{corr}$-correct if the final key $\tilde{K}_{A}$ of Alice differs from the final key $\tilde{K}_{B}$ of Bob with probability less than $\epsilon_{corr}$, i.e.
		\begin{equation}
		\mathrm{Pr}(\Tilde{K}_{A} \neq \Tilde{K}_{B}) \leq \epsilon_{corr} \, .
		\end{equation}
	\end{defn}
	\begin{defn}[secrecy\cite{renner2008security,arnon2019simple,murta2019towards}]\label{Secrecy}
		For any $\epsilon_{sec} \geq 0$, a DIQKD protocol is $\epsilon_{sec}$ w.r.t the adversary E if the joint state satisfies 
		\begin{equation}
		p(\Omega) \frac{1}{2} \|\rho_{\Tilde{K}_{A}E|\Omega}-\tau_{\Tilde{K}_{A}} \otimes \rho_{E} \|_{1} \leq \epsilon_{sec} \, ,
		\end{equation}
		where $\tau_{\Tilde{K}_{A}}$ is the maximally mixed state on $\Tilde{K}_{A}$ of the protocol. Here $p(\Omega)$ is the probability of not aborting the protocol.
	\end{defn}
	If a protocol is $\epsilon_{corr}$-correct and $\epsilon_{sec}$-secret, then it is $\epsilon_{DIQKD}^s$-correct and secret for any $\epsilon_{DIQKD}^s \geq \epsilon_{corr}+\epsilon_{sec}$.\\
	The correctness (see Def. \ref{Correctness}) of the final key is ensured by the error correction step. During error correction, Alice sends a sufficient amount of information to Bob so that he can correct his raw key. If Alice and Bob do not abort in this step, then the probability that they end up with different raw keys is guaranteed to be very small (below $\epsilon_{EC}$). For the secrecy of the protocol (see Def. \ref{Secrecy}), one needs to estimate how far the final state describing Alice’s key and the eavesdropper's system is from the ideal one.

	\begin{defn}[Secret key rate\cite{arnon2019simple,murta2019towards}]
		If a protocol generates a correct and secret key of length l after n rounds, the the secret key rate is defined as 
		\begin{equation}
		r=\frac{l}{n} \, .
		\end{equation}
	\end{defn} 
	
	Any useful DIQKD protocol should not abort almost all the time. This is apprehended by the concept of completeness.
	\begin{defn}[security\cite{arnon2019simple,murta2019towards}]\label{security defn}
		A DIQKD protocol is $(\epsilon_{DIQKD}^s,\epsilon_{DIQKD}^c,l)$-secure if 
		\begin{enumerate}
			\item (soundness) For any implementation of the protocol, either it aborts with probability greater than $1-\epsilon_{DIQKD}^s$ or an $\epsilon_{DIQKD}^s$-correct and secret key of length l is obtained.
			\item (Completeness) There exists an honest implementation of the protocol such that the probability of not aborting, $p(\Omega)$, is greater than $1-\epsilon_{DIQKD}^c$.
		\end{enumerate}
	\end{defn}
	In the privacy amplification step, Alice and Bob want to turn their equal string of bits, which may be partially known to an eavesdropper, into a shorter completely secure string of bits. For this step, a 2-universal family of hash functions is needed. 
	\begin{defn}[2-universal hash function]\label{2-universal hash function}
	A family of hash functions $\mathcal{F}=\{f:$ $\left.\{0,1\}^{n} \rightarrow\{0,1\}^{\ell}\right\}$ is called 2-universal if for every two strings $x, x^{\prime} \in\{0,1\}^{n}$ with $x \neq x^{\prime}$ then
	\begin{equation}
	    \operatorname{Pr}_{f \in \mathcal{F}}\left(f(x)=f\left(x^{\prime}\right)\right)=\frac{1}{2^{\ell}} \, ,
	\end{equation}
	where $f$ is chosen uniformly at random in $\mathcal{F}$. The property of 2-universality ensures a good distribution of the outputs. For $\ell \leq n$ there always exist a 2-universal family of hash functions  \cite{carter1979universal}.
	\end{defn}
	Now we will state the quantum Leftover Hashing Lemma \cite{renner2005simple,tomamichel2011leftover}. It quantifies the secrecy of a protocol as a function of a conditional entropy of the state before privacy amplification and the length of the final key.
	\begin{theorem}[Leftover Hashing Lemma with smooth min-entropy\cite{arnon2019simple,murta2019towards,tomamichel2011leftover}]\label{leftover hashing}
		Let $\rho_{A^nE}$ be a classical quantum state. Let $\mathcal{F}$ be a 2-universal family of hash functions, from $\{0,1\}^n$ to $\{0,1\}^l$, that maps the classical n-bit string $A^{n}$ into $K_A$. Then
		\[ \| \rho_{K_AFE}-\tau_{K_A}\otimes \rho_{FE} \|  \le 2^{-\frac{1}{2}(H_{\text{min}}^{\epsilon}(A^n|E)_{\rho}-l)}+2\epsilon \, ,   \]
		where $F$ is a classical register that stores the hash function $f$.
	\end{theorem}
	With the Leftover hash lemma and the definition of secrecy (see Def. \ref{Secrecy}), we express the length of a secure key as a function of the entropy of Alice's raw key conditioned on Eve's information before privacy amplification.
	\begin{theorem}[Key length\cite{arnon2019simple,murta2019towards}]\label{key length}
		Let $P(\Omega)$ be the probability that the DIQKD protocol does not
		abort for a particular implementation. If the length of the key generated after privacy amplification is given by
		\[ l  \leqslant H_{\text{min}}^{\epsilon_{s}/P(\Omega)}(A^n|E)_{\rho_{|\Omega}}- 2\log\frac{1}{2\epsilon_{PA}} \, , \]
		then the DIQKD protocol is $\epsilon_{PA}+\epsilon_{s}$-secret.
	\end{theorem}
	In this paper, we have considered the IID scenario (collective attacks). In the assumption of collective attacks, the distributed state and the behavior of Alice's and Bob's devices are the same in every round of the protocol. Eve can carry out arbitrary operations in her quantum side information. This assumption implies that after $n$ rounds of the protocol, the state shared by Alice, Bob and Eve is $\rho_{A^n B^n E}=\rho_{ABE}^{\otimes n}$. The quantum asymptotic equipartition property \cite{tomamichel2009fully,tomamichel2015quantum} allows to bound the conditional smooth min-entropy of state $\rho_{AE}^{\otimes}$ by the conditional von Neumann entropy of the state $\rho_{AE}$.
	\begin{theorem}[Asymptotic equipartition property \cite{tomamichel2009fully}]\label{Asymptotic equipartition property}
		Let $\rho=\rho_{AE}^{\otimes n}$ be an IID state. Then for $ n \ge \frac{8}{5}\log\frac{2}{\epsilon^2} $
		\[ H_{\text{min}}^{\epsilon} (A^n \arrowvert E^n)_{\rho_{AE}^{\otimes n}} > nH(A \arrowvert E)_{\rho_{AE}} - \sqrt{n}\delta(\epsilon,\chi) \, , \]
		and similarly
		\[ H_{\text{max}}^{\epsilon} (A^n \arrowvert E^n)_{\rho_{AE}^{\otimes n}} < nH(A \arrowvert E)_{\rho_{AE}} + \sqrt{n}\delta(\epsilon,\chi) \, , \]
		where $ \delta(\epsilon,\chi)=4\log (\chi) \sqrt{\log \frac{2}{\epsilon^2}} $ and $ \chi = \sqrt{2^{-H_{min}(A\arrowvert E)_{\rho_{AE}}}}+\sqrt{2^{H_{max}(A\arrowvert E)_{\rho_{AE}}}}+1 $.
	\end{theorem}
	
	\begin{lemma}\label{QBER_lemma} \cite{grasselli2019conference,yin2019finite}
		Let $\mathcal{X}_{n+k}$ be a random binary string of $ n + k $ bits, $\mathcal{X}_{k}$ be a random sample (without replacement) of m entries from the string $\mathcal{X}_{n+k}$ and $\mathcal{X}_{n}$ be the remaining bit string. $ \Lambda_{k} $ and $ \Lambda_{n} $ are the frequencies of bit value 1 in string $\mathcal{X}_{k}$ and $\mathcal{X}_{n}$, respectively. For any $\varepsilon_1 > 0 $, it holds the upper tail inequality:
		\begin{equation}
		\mathrm{Pr}[\Lambda_{n} \ge \Lambda_{k}+\gamma_1(n,k,\Lambda_{k},\varepsilon_1)] > \varepsilon_1 \, ,
		\end{equation} 
		where $ \gamma_1(a,b,c,d) $ is the positive root of
		\begin{equation*}
		\ln{\binom{bc}{b}}+\ln{ \binom{ac+a\gamma_1(a,b,c,d)}{a} } = \ln{\binom{(a+b)c+a\gamma_1(a,b,c,d)}{a+b}} + \ln{d} \, .
		\end{equation*}
		For $ \varepsilon_2 > 0 $, we have the lower tail inequality:
		\begin{equation}
		\mathrm{Pr}[\Lambda_{n} \le \Lambda_{k}-\gamma_2(n,k,\Lambda_{k},\varepsilon_2)] > \varepsilon_2 \, ,
		\end{equation} 
		where $ \gamma_2(a,b,c,d) $ is the positive root of
		\begin{equation*}
		\ln{\binom{bc}{b}}+\ln{ \binom{ac-a\gamma_2(a,b,c,d)}{a} } = \ln{\binom{(a+b)c-a\gamma_2(a,b,c,d)}{a+b}} + \ln{d} \, .
		\end{equation*}
	\end{lemma}
	
	\begin{lemma}\label{Hoeffding lemma} \cite{hoeffding1963large,hoeffding1994probability}
		Let $X_1, X_2, \cdots, X_n$ be independent random variables strictly bounded by the intervals $[a_i,b_i]$, i.e. $a_i \le X_i \le b_i$. We define 
		\begin{equation*}
		\Bar{X}=\frac{1}{n}(X_1+X_2+\cdots+X_n) \, .
		\end{equation*} 
		Then, Hoeffding's inequality reads
		\begin{equation*}
		\mathrm{Pr}\Big(\Bar{X}-E[\Bar{X}] \ge t \Big) \le \exp{\Bigg(-\frac{2n^2t^2}{\sum_{i=1}^n(b_i-a_i)^2}}\Bigg) \, .
		\end{equation*}
		Let $c_i:=b_i-a_i$ and $c_i \le C$ $\forall$ $i$. Then, Hoeffding's inequality reads 
		\begin{equation*}
		\mathrm{Pr}\Big(\Bar{X}-E[\Bar{X}] \ge t \Big) \le \exp{\Bigg(-\frac{2n^2t^2}{nC^2}}\Bigg) =  \exp{\Bigg(-\frac{2nt^2}{C^2}}\Bigg)\, .
		\end{equation*}
	\end{lemma}

	\section{Secret key analysis} \label{Security key analysis}
	
	\begin{theorem}[Completeness]
		The DIQKD protocol stated in Sec. \ref{protocol} is $\epsilon_{est}+\epsilon^{\gamma}_{est}$ complete.
	\end{theorem}
	\begin{proof}
		The protocol can abort in two instances. Either it will abort if the error correction failed or if the estimated Bell violation $B[\hat{\textbf{P}}_3]$ is not high enough. The probability that the error correction fails can only happen if the real QBER $ Q $ is larger than $ \hat{Q}+\gamma_{est} $, which happens with probability $\epsilon_{est}^{\gamma}$, see Sec. \ref{protocol} for details. The protocol also aborts if the estimated Bell violation $B[\hat{\textbf{P}}_3]$ is smaller $B[\hat{\textbf{P}}_2]-\delta_{est})$, see Sec. \ref{protocol} for details. Thus, considering an honest implementation consisting of IID rounds, we can bound the probability of abortion of the protocol: 
		\begin{align}\label{Appendix equation 1}
		\text{p(abort)} &= \text{p((EC aborts) or (Bell test fails))} \\ \nonumber
		& \le \text{p(EC aborts)}+\text{p(Bell test fails)} \\ \nonumber
		& \le \text{p(QBER test fails)}+\text{p(Bell test fails)} \\ \nonumber
		& = p(Q > \hat{Q} + \gamma_{est}) + p(B[\hat{\textbf{P}}_3] < B[\hat{\textbf{P}}_2]-\delta_{est}) \\ \nonumber
		& = \epsilon^{\gamma}_{est} + \epsilon_{est} \, ,
		\end{align}	
		where $ \epsilon_{est} $ is defined in Eq.~(\ref{BellVal deviation ineq}), and $ \epsilon^{\gamma}_{est} $ is defined in Eq.~(\ref{QBER deviation ineq}). Thus, we get $\epsilon_{DIQKD}^c \leq \epsilon_{est}+\epsilon^{\gamma}_{est}$.
	\end{proof}
	
	For the \textbf{soundness}, we have to evaluate the correctness and secrecy, defined in Def. \ref{Correctness} and Def. \ref{Secrecy}, respectively. In case of correctness, if we have an error correction protocol that does not abort, then Alice (Bob) will have the raw key $K_A$ ($K_B$) after the protocol. The string $K_B$ differs from $K_A$ with probability less than $\epsilon_{EC}$ and as the final keys $\Tilde{K}_A$ and $\Tilde{K}_B$ are equal if the raw keys are equal, it follows:
	\begin{equation*}
	P(\Tilde{K}_A \neq \Tilde{K}_B) \le P(K_A \neq K_B) \le \epsilon_{EC} \, .
	\end{equation*}
	For secrecy, let us recall that $\Omega$ is defined as the event when the protocol does not abort. This happens when the error correction protocol does not abort and achieved the required Bell violation according to Alice's and Bob's outputs (and inputs). Now define the event $\hat{\Omega}$ as the event $\Omega$ (protocol not aborting) and the error correction being successful i.e. $K_A=K_B$. Thus, 
	\begin{align}\label{state discrimination of EC}
	\|\rho_{\Tilde{K}_{A}E_{|\Omega}}-\tau_{\Tilde{K}_{A}}\otimes\rho_{E}\|_{1} & \le \| \rho_{\Tilde{K}_{A}E_{|\Omega}}-\rho_{\Tilde{K}_{A}E_{|\hat{\Omega}}} \|_{1} + \| \rho_{\Tilde{K}_{A}E_{|\hat{\Omega}}} - \tau_{\Tilde{K}_{A}}\otimes\rho_{E}\|_{1} \\ \nonumber
	& \le \epsilon_{EC}+ \| \rho_{\Tilde{K}_{A}E_{|\hat{\Omega}}} - \tau_{\Tilde{K}_{A}}\otimes\rho_{E}\|_{1} \, .
	\end{align}
	The first inequality follows from the triangular inequality of the trace distance \cite{nielsen2002quantum}. $\rho_{\Tilde{K}_{A}E_{|\Omega}}$ is the joint classical quantum state of Alice and Eve if the protocol does not abort. $\rho_{\Tilde{K}_{A}E_{|\hat{\Omega}}}$ is the joint classical quantum state of Alice and Eve if the protocol does not abort and the error correction is successful. When error correction succeeds, the probability of $ K_A=K_B $ is higher than $ (1-\epsilon_{EC}) $. Conversely, the probability $K_A \neq K_B$ is less than $\epsilon_{EC}$. Thus the second inequality of the Eq.~(\ref{state discrimination of EC}) comes from 
	\begin{align}
	    \| \rho_{\Tilde{K}_{A}E_{|\Omega}}-\rho_{\Tilde{K}_{A}E_{|\hat{\Omega}}} \|_{1} \le  (1-\epsilon_{EC}) \| \rho_{\Tilde{K}_{A}E_{|\hat{\Omega}}} - \rho_{\Tilde{K}_{A}E_{|\hat{\Omega}}} \|_{1} + \epsilon_{EC} \| \rho_{\Tilde{K}_{A}E_{|\hat{\Omega}}} - \rho_{\Tilde{K}_{A}E_{|\hat{\Omega}^c}} \|_{1} \le \epsilon_{EC} 
	\end{align}
	$\hat{\Omega}^{c}$ is defined as the event when the protocol does not abort but error correction is not successful , i.e $K_A \neq K_B$.\\
	
	Now we proceed to evaluate the term $\| \rho_{\Tilde{K}_{A}E_{|\hat{\Omega}}} - \tau_{\Tilde{K}_{A}}\otimes\rho_{E}\|_{1}$ of Eq.~(\ref{state discrimination of EC}). We will follow the path showed in \cite{murta2019towards,arnon2019simple}. Given that the protocol did not abort, the maximal length of a secure key is determined by the smooth min-entropy of Alice's raw key conditioned on all information available to the eavesdropper (See the leftover hashing lemma in Theorem. \ref{leftover hashing}). In our protocol (see Sec. \ref{protocol}), it is given by $H_{\text{min}}^{\epsilon_s}(A^N|X^N Y^N T^N E O_{EC})_{\rho_{|\hat{\Omega}}}$. Here we recall that $O_{EC}$ is the information exchanged by Alice and Bob during the error correction protocol. $X^N$ and $Y^N$ are the input bit strings (measurement settings) for Alice and Bob, respectively. $A^N$ is the output bit string of Alice. $T^N$ is the shared random key that determines whether the round is a test or a key generation round. $\hat{\Omega}$ is the event that the protocol does not abort and error correction succeeds.\\
	
	In order to bypass the conditioned state of $H_{\text{min}}^{\epsilon_s}(A^N|X^N Y^N T^N E O_{EC})_{\rho_{|\hat{\Omega}}}$, we can start from the definition of Secrecy (see Def. \ref{Secrecy}). Then we have to bound the term 
	\begin{equation}
	p(\Omega)\| \rho_{\Tilde{K}_{A}FE_{|\Omega}} - \tau_{\Tilde{K}_{A}}\otimes\rho_{FE}\|_{1} = \| \rho_{\Tilde{K}_{A}FE\land \Omega} - \tau_{\Tilde{K}_{A}}\otimes\rho_{FE \land \Omega}\|_{1} \, ,
	\end{equation}
	where $\rho_{\Tilde{K}_{A}FE\land \Omega} = p(\Omega)\rho_{\Tilde{K}_{A}FE_{|\Omega}}$ is a subnormalized state. Here we recall that $F$ is the classical register that stores the hash function $f$ (see Def. \ref{2-universal hash function}).\\
	
	Now using the leftover hashing lemma in Theorem \ref{leftover hashing}, we can generate an $(\epsilon_{s}+\epsilon_{PA})$-secret key of length \cite{murta2019towards}
	\begin{equation}
	l  \leqslant H_{\text{min}}^{\epsilon_{s}}(A^N|E)_{\rho \land \Omega}- 2\log\frac{1}{2\epsilon_{PA}} \, .
	\end{equation}
    In Ref.\cite{tomamichel2017largely}, it is proved that 
    \begin{equation}
        H_{\text{min}}^{\epsilon_{s}}(A^N|E)_{\rho \land \Omega} \ge H_{\text{min}}^{\epsilon_{s}}(A^N|E)_{\rho} \, .
    \end{equation}
	Thus, we proceed to estimate the quantity $H_{\text{min}}^{\epsilon_s}(A^N|X^N Y^N T^N E O_{EC})_{\rho}$ in order to bound the achievable secret key of length $l$. \\
	Using the chain rule relation for the smooth min-entropy conditioned on classical information \cite{tomamichel2015quantum}, we can write
	\begin{equation}
	H_{\text{min}}^{\epsilon_s}(A^N|X^N Y^N T^N E O_{EC})_{\rho} = H_{\text{min}}^{\epsilon_s}(A^N|X^N Y^N T^N E)_{\rho}-\text{leak}_{EC} \, .
	\end{equation}
	Thus, in order to bound $H_{\text{min}}^{\epsilon_s}(A^N|X^N Y^N T^N E O_{EC})_{\rho}$, we have to lower bound $H_{\text{min}}^{\epsilon_s}(A^N|X^N Y^N T^N E)_{\rho}$ and upper bound $\text{leak}_{EC}$ (the leakage due to the error correction). 

	\subsection{Estimation of $\text{leak}_{\text{EC}}$}\label{Estimation of leak}
	Alice and Bob perform an EC procedure so that Bob can compute a guess of Alice's raw key $A^N$. In order to verify if EC is successful, Alice chooses a two-universal hash function (uniformly at random) from the family of hash functions and computes a hash of length $\log(\frac{1}{\epsilon_{EC}})$ from her raw keys $A^N$. Then she sends the chosen hash function and the hashed value of her bits to Bob via a public channel. We denote all the classical communication (information leaked during EC, hash function and the hashed value for verification) by $O_{EC}$. Bob computes the hash function on his key. If the hashed values are equal, then Alice’s and Bob’s keys are the same with high probability. If the hashed values are different, the parties will abort the protocol. During this whole process, the amount of information about the key exposing to the adversary Eve is termed as $\text{leak}_{EC}$. In Ref.\cite{renner2005simple}, the $\text{leak}_{EC}$ is bounded by
	\begin{equation} \label{Leak Eq.1}
	\text{leak}_{EC} \leq H_0^{\epsilon'_{EC}}(A^N|B^N X^N Y^N T^N ) + \log\frac{1}{\epsilon_{EC}} \, ,
	\end{equation} 
	where $\epsilon_{EC}^c=\epsilon_{EC}+\epsilon'_{EC}$ (see Table.\ref{DIQKD_protocol_parameter}). $H_0$ is the R\'{e}nyi entropy introduced in Ref.\cite{renner2005simple}. In Ref.\cite{tomamichel2015quantum}, it is denoted as $\Bar{H}_0^{\uparrow}$.	If Alice and Bob do not abort, then their resultant bit string is identical ($ K_A=K_B $) with at least $ 1-\epsilon_{EC} $ probability. We can bound the entropy $ H_0^{\epsilon'_{EC}}(A^N|B^N X^N Y^N T^N ) $ in the following way:
	\begin{align} \label{Leak Eq.2}
	H_0^{\epsilon'_{EC}}(A^N|B^N X^N Y^N T^N ) &\leq H_{max}^{ \frac{\epsilon'_{EC}}{2} }(A^N|B^N X^N Y^N T^N )  +  \log \bigg(\frac{8}{{\epsilon'}^{2}_{EC}} + \frac{2}{2-{\epsilon'}_{EC}} \bigg) \\ \nonumber
	& \leq NH(A|BXYT)+4\sqrt{N}\log(2\sqrt{2^{\log_{2}d}}+1)\sqrt{\log \frac{8}{{\epsilon'_{EC}}^{2}}} + \\ \nonumber
	& \qquad \log \bigg(\frac{8}{{\epsilon'}^{2}_{EC}} + \frac{2}{2-{\epsilon'}_{EC}} \bigg) \, .
	\end{align}
	For the definition of $H_{0}^{\epsilon}(A|B)$, see Ref.\cite{renner2005simple}. The first inequality of Eq.~(\ref{Leak Eq.2}) comes from Ref.\cite{tomamichel2011leftover} and the Eq.~(B11) of Ref.[\cite{murta2019towards}]. The last inequality comes from the asymptotic equipartition property (see Theorem. \ref{Asymptotic equipartition property}), where we used the relations 
	\begin{align}
	    \delta(\epsilon,\chi) &= 4\log (\chi) \sqrt{\log \frac{2}{\epsilon^2}} \\ \nonumber
	    & \le 4\log(2 \sqrt{ 2^{\log_{2}d} }+1)\sqrt{\log(\frac{2}{\epsilon^2})} \, .
	\end{align}
	Here we have used $\chi \le 2 \sqrt{ 2^{\log_{2}d} }+1 $ which comes from
	\begin{align} \label{max-entropy log bound}
	\chi &= \sqrt{2^{-H_{min}(A\arrowvert E)_{\rho_{AE}}}}+\sqrt{2^{H_{max}(A\arrowvert E)_{\rho_{AE}}}}+1 \\ \nonumber
	 & \le 2 \sqrt{2^ { H_{max}(A \arrowvert XYTE)_{\rho}}}+1 \\ \nonumber
	 & \le 2 \sqrt{ 2^{\log_{2}d} }+1
	\end{align}
	The first inequality of Eq.~(\ref{max-entropy log bound}) follows from the fact that A is a classical register, and therefore has positive conditional min-entropy, which implies $ -H_{min}(A \arrowvert XYTE) \le H_{min}(A \arrowvert XYTE) \le H_{max}(A \arrowvert XYTE) $. For the second inequality of Eq.~(\ref{max-entropy log bound}), we use $ H_{max}(A \arrowvert XYTE) \le \log_{2}d $. 
	
	Therefore, from Eq.~(\ref{Leak Eq.1}) and Eq.~(\ref{Leak Eq.2}), we can bound the leakage in the following way
	\begin{align} \label{Leak Eq.3}
	\text{leak}_{EC} & \leq NH(A|BXYT)+\sqrt{n}(4\log(2\sqrt{2^{\log_{2}d}}+1))\sqrt{\log \frac{8}{{\epsilon'}^{2}_{EC}}} + \\ \nonumber
	& \qquad \log \bigg(\frac{8}{{\epsilon'}^{2}_{EC}} + \frac{2}{2-{\epsilon'}_{EC}} \bigg)+\log\frac{1}{\epsilon_{EC}} \, .
	\end{align}
	Now we bound the single round von-Neumann entropy $ H(A|BXYT) $ as
	\begin{align} \label{Leak Eq.4}
	H(A|BXYT) &= p(T = 0)H(A|BXY T = 0)+ p(T = 1)H(A|BXY T = 1)\\ \nonumber
	& \leq  (1-\xi)H(A|BXY T = 0)+\xi \log_{2}d	\\ \nonumber
	& \leq  (1-\xi-\eta)H(A|BXY T = 0)+(\xi+\eta) \log_{2}d \, .
	\end{align}
	See Table. \ref{DIQKD_protocol_parameter} for the details of $\xi$, $\eta$ and $\gamma_{est}$. For the first equality, we have used that for the conditional von Neumann entropy, it holds $H(A|BX)_{\rho}=\sum_{x}p(X=x)H(A|BX=x)$. We divide the measurement rounds into key generation (specified by $T=0$) and parameter estimation (specified by $T=1$), for details see Sec. \ref{protocol}. The first inequality comes from the fact that parameter estimation round's measurements were publicly communicated to estimate the Bell inequality and the corresponding violation. $\eta$ rounds of the raw key generation measurement were communicated through a public channel to estimate the QBER which leads to the last inequality. \\
	Now our goal is to estimate $H(A|BXY T = 0)$. For dichotomic observables and uniform marginals, $ H(A \arrowvert B) $ can be expressed as $ h(Q) $ \cite{pironio2009device}, where $ h $ is the binary entropy function, $ h(p):=-p\log_{2}p-(1-p)\log_{2}(1-p) $. Similarly for the $ [(m_a,m_b),d] $ Bell scenario, $ H(A \arrowvert B) $ can be expressed as a function of the QBER, $ H(A \arrowvert B) = -Q\log_{2}Q-(1-Q)\log_{2}(1-Q)+Q\log_{2}(d-1) $ \cite{bradler2016finite}. \\
	For our specific protocol (see Sec. \ref{protocol}), we bound $H(A|BXY T = 0)$ by a function of $ \hat{Q}_1 + \gamma_{est} $ (observed QBER + estimated statistical error), see Sec. \ref{BellVal_skr} for details:
	\begin{equation}\label{QBER bound}
	H(A|BXY,T=0) \le f(\hat{Q} + \gamma_{est}) \, ,
	\end{equation}
	where $f(x)=h(x)+x\log_{2}(d-1)$ ($ d $ is the number of outcomes per measurement in the Bell scenario) and $ h $ is the binary entropy function. From Eq.~(\ref{Leak Eq.4}) and Eq.~(\ref{QBER bound}), then it follows:
	\begin{equation} \label{Leak Eq.5}
	    H(A|BXYT) \leq (1-\xi-\eta) f(\hat{Q} + \gamma_{est}) + (\xi+\eta) \log_{2}d  \, .
	\end{equation}
	
	The leakage due to error correction is given by (from Eq.~(\ref{Leak Eq.3}) and Eq.~(\ref{Leak Eq.5})) 
	\begin{align}\label{leakage bound}
	\text{leak}_{EC} & \leq N[(1-\xi-\eta) f(\hat{Q} + \gamma_{est})+ (\xi+\eta) \log_{2}d] + \sqrt{N}\bigg(4\log(2\sqrt{2^{\log_{2}d}}+1)\bigg)\sqrt{\log \frac{8}{{\epsilon'}^{2}_{EC}}}+ \\ \nonumber
	& \qquad \qquad \log \bigg( \frac{8}{{\epsilon'}^{2}_{EC}} + \frac{2}{2-{\epsilon'}_{EC}} \bigg) +\log \frac{1}{\epsilon_{EC}} \, .
	\end{align}
	
	\subsection{Estimation of min-entropy $H_{\text{min}}^{\epsilon_s}(A^N|X^N Y^N T^N E)_{\rho}$} \label{Estimation of min-entropy}
    Finally, we lower bound $H_{\text{min}}^{\epsilon_s}(A^N|X^N Y^N T^N E)_{\rho}$. We use the asymptotic equipartition property (see Theorem \ref{Asymptotic equipartition property}) to lower bound the min-entropy of $N$ rounds by the von Neumann entropy of single rounds:
	\begin{align}
	H_{\text{min}}^{\epsilon_s}(A^N|X^N Y^N T^N E)_{\rho} &\geqslant NH(A|XYTE)_{\rho}- 4\sqrt{N}\log(2\sqrt{2^{\log_{2} d}}+1) \sqrt{\log\frac{2}{\epsilon_s^2}} \, .
	\end{align}
	Since, Alice's actions (and her device's) are independent of Bob’s choice of input, adding information about $ Y $ (Bob's input) does not increase (or decrease) the conditional von Neumann entropy $H(A|X,E)_{\rho}$. Since $H(A|X,E)_{\rho}$ and $H(A|XYE, T=1)_{\rho}$ are equivalent in our set-up, we will use both terms interchangeably. In the general scenario, the conditional von-Neumann entropy is hard to calculate analytically. But the conditional von Neumann entropy can be lower bounded by the conditional min-entropy as
	\begin{equation} \label{von-Neumann min-entropy reln}
	H(A|XYT,E)_{\rho} \geq H_{min}(A|XYT,E)_{\rho} \, .
	\end{equation}
	The advantage of looking at the conditional min-entropy is that we can express it as $H_{min}(A|XYE,T=1)_{\rho}=-\log_{2}P_g(A|X,E)$ \cite{konig2009operational}, where $ P_g(A|X,E) $ is Eve's guessing probability about Alice's $ X $-measurement results $ A $ conditioned on her side information $ E $. $ P_g(A|X,E) $ can be upper bounded by a function $ G_x $ of the expected Bell violation $B[\textbf{P}]$ \cite{masanes2011secure} by solving a semi-definite programme \cite{johnston2016qetlab}, i.e. $ P_{g}(A|X,E) \leq G_x(B[\textbf{P}]) $. For our specific protocol (see Sec. \ref{protocol}), we will lower bound the min-entropy (via upper bounding the guessing probability $ P_g(A|X,E) $) using the Bell inequality $ B $ and corresponding Bell value $ B[\hat{\textbf{P}}_2]-\delta_{est}-\delta_{con} $ (explained in Sec.\ref{BellVal_skr}): 
	\begin{equation} \label{entropy bound}
	H_{\text{min}}(A|XYE, T=1)_{\rho} \geqslant -\log_{2} G_{x} (B[\hat{\textbf{P}}_2]-\delta_{est}-\delta_{con}) \, .
	\end{equation} \\

	Finally, putting Eq.~(\ref{leakage bound}) and Eq.~(\ref{entropy bound}) together, we have either the protocol mentioned in Sec. \ref{protocol} aborts with probability higher than $1-(\epsilon_{con}+\epsilon^{c}_{EC})$ or a $(2\epsilon_{EC} + \epsilon_{s} + \epsilon_{PA})$)-correct and secret key can be generated of length $l$. The length $l$ is bounded by 
	\begin{align}
	l &\leqslant N\Bigg[ -\log_{2} G_{x} (B[\hat{\textbf{P}}_2]-\delta_{est}-\delta_{con}) -  (1-\xi-\eta) f(\hat{Q} + \gamma_{est}) - (\xi+\eta) \log_{2}d \Bigg]- \\ \nonumber
	& \sqrt{N}\bigg( 4\log \big(2\sqrt{2^{\log_{2} d}}+1 \big) \bigg(\sqrt{\log\frac{8}{{\epsilon'}^{2}_{EC}}}+\sqrt{\log\frac{2}{\epsilon_s^2}}\bigg)\bigg) - \log\bigg(\frac{8}{{\epsilon'}^{2}_{EC}} + \frac{2}{2-{\epsilon'}_{EC}}\bigg)-\log \frac{1}{\epsilon_{EC}}-2\log\frac{1}{2\epsilon_{PA}} \, .
	\end{align}
	
	\section{Measurement Settings}\label{imp meas settings}
	Here we list the explicit measurement settings employed in Sec. \ref{Results 2m2}. 
	\begin{equation}\label{3-meas settings_1 appendix}
	\begin{aligned}
	x &=1 \Rightarrow \sigma_z \, ,  
	&
	y &=1 \Rightarrow \begin{bmatrix}
	0.7064 &  -0.6632 + 0.2473i \\
	-0.6632 - 0.2473i & -0.7064
	\end{bmatrix}\, , \\
	x &=2 \Rightarrow \begin{bmatrix}
	-0.1817 &	0.1307 + 0.9746i \\
	0.1307 - 0.9746i &	0.1817
	\end{bmatrix}\, ,
	& 
	y &=2 \Rightarrow \begin{bmatrix}
	-0.6882 & -0.2128 - 0.6936i \\
	-0.2128 + 0.6936i & 0.6882
	\end{bmatrix}\, , \\
	x &=3 \Rightarrow \begin{bmatrix}
	-0.7746 & 0.6186 - 0.1315i \\
	0.6186 + 0.1315i & 0.7746
	\end{bmatrix}\, ,
	&
	y &=3 \Rightarrow \begin{bmatrix}
	0.4046 & -0.1960 + 0.8932i \\
	-0.1960 - 0.8932i  & -0.4046
	\end{bmatrix}\, .
	\end{aligned}
	\end{equation}
	Using the following set of measurement settings for Alice and Bob in Eq.~(\ref{3-meas settings_1 appendix}), one can generate a higher secret key rate employing our method than using any subset of two measurement settings per party using the standard CHSH inequality.
	
	\begin{equation}\label{3-meas settings_2 appendix}
	\begin{aligned}
	x &=1 \Rightarrow \sigma_z \, , 
	&
	y &=1 \Rightarrow \begin{bmatrix}
	-0.4091 &  -0.5937 + 0.6930i \\
	-0.5937 - 0.6930i &  0.4091
	\end{bmatrix}\, , \\
	x &=2 \Rightarrow \begin{bmatrix}
	0.7019 &	0.5167 - 0.4903i \\
	0.5167 + 0.4903i &	-0.7019
	\end{bmatrix}\, ,
	&
	y &=2 \Rightarrow \begin{bmatrix}
	-0.6133 &  -0.2514 + 0.7488i \\
	-0.2514 - 0.7488i & 0.6133
	\end{bmatrix}\, .
	\end{aligned}
	\end{equation}
	
	Using the following measurement settings in Eq.~(\ref{3-meas settings_2 appendix}) and the state in Eq.~(\ref{BellState}) with no white noise, one cannot extract a secret key using our method or blindly using the CHSH inequality. 
	\begin{equation}\label{add measurement appendix}
	\begin{aligned}
	x &=3 \Rightarrow \begin{bmatrix}
	-0.1457 & -0.9777 + 0.1513i \\
	-0.9777 - 0.1513i &  0.1457
	\end{bmatrix}\, ,
	&
	y &=3 \Rightarrow \begin{bmatrix}
	-0.9020 & -0.3795 - 0.2056i\\
	-0.3795 + 0.2056i  &  0.9020
	\end{bmatrix}\, .
	\end{aligned}
	\end{equation}
	However, by adding another set of measurements for Alice and Bob mentioned in Eq.~(\ref{add measurement appendix}), it is possible to achieve a non-zero secret key rate using our method. 
	
	\section{Tabular representation of Bell inequality}\label{APP CG notation}
	Here we introduce an alternative representation of the hyperplane vector (see Eq.~(\ref{hyperplane})). We rearrange the entries (coefficients of the Bell inequality) in a tabular construction. For the $[2,2]$ Bell scenario, it is represented in Table \ref{CHSH hyperplane table [2,2]}. 
		\begin{table}[h]
			\centering
			\begin{tabular}{p{1cm}  p{1cm} | p{1cm}  p{1cm}}
				$h_{A_1B_1}^{11}$ & $h_{A_1B_1}^{12}$ & $h_{A_1B_2}^{11}$ & $h_{A_1B_2}^{12}$ \\ 
				$h_{A_1B_1}^{21}$ & $h_{A_1B_1}^{22}$ & $h_{A_1B_2}^{21}$ & $h_{A_1B_2}^{22}$ \\  
				\hline 
				$h_{A_2B_1}^{11}$ & $h_{A_2B_1}^{12}$ & $h_{A_2B_2}^{11}$ & $h_{A_2B_2}^{12}$ \\  
				$h_{A_2B_1}^{21}$ & $h_{A_2B_1}^{22}$ & $h_{A_2B_2}^{21}$ & $h_{A_2B_2}^{22}$ 
			\end{tabular}
			\caption{Bell inequality table for the [2,2] scenario.}
			\label{CHSH hyperplane table [2,2]}
		\end{table}\\
		This representation is used in Table \ref{CHSH facet}. Similarly, we reorder the elements of the hyperplane vector for the [2,3] Bell scenario in the following way (see Table \ref{CGLMP hyperplane table [2,3]}):
		\begin{table}[h]
			\centering
			\begin{tabular}{p{1cm}  p{1cm} p{1cm} | p{1cm}  p{1cm} p{1cm} }
				$h_{A_1B_1}^{11}$ & $h_{A_1B_1}^{12}$ & $h_{A_1B_1}^{13}$ & $h_{A_1B_2}^{11}$ & $h_{A_1B_2}^{12}$ & $h_{A_1B_2}^{13}$ \\  
				$h_{A_1B_1}^{21}$ & $h_{A_1B_1}^{22}$ & $h_{A_1B_1}^{23}$ & $h_{A_1B_2}^{21}$ & $h_{A_1B_2}^{22}$ & $h_{A_1B_2}^{23}$ \\ 
				$h_{A_1B_1}^{31}$ & $h_{A_1B_1}^{32}$ & $h_{A_1B_1}^{33}$ & $h_{A_1B_2}^{31}$ & $h_{A_1B_2}^{32}$ & $h_{A_1B_2}^{33}$ \\ 
				\hline 
				$h_{A_2B_1}^{11}$ & $h_{A_2B_1}^{12}$ & $h_{A_2B_1}^{13}$ & $h_{A_2B_2}^{11}$ & $h_{A_2B_2}^{12}$ & $h_{A_2B_2}^{13}$ \\ 
				$h_{A_2B_1}^{21}$ & $h_{A_2B_1}^{22}$ & $h_{A_2B_1}^{23}$ & $h_{A_2B_2}^{21}$ & $h_{A_2B_2}^{22}$ & $h_{A_2B_2}^{23}$ \\
				$h_{A_2B_1}^{31}$ & $h_{A_2B_1}^{32}$ & $h_{A_2B_1}^{33}$ & $h_{A_2B_2}^{31}$ & $h_{A_2B_2}^{32}$ & $h_{A_2B_2}^{33}$ 
			\end{tabular}
			\caption{Bell inequality table for the [2,3] scenario.}
			\label{CGLMP hyperplane table [2,3]}
	\end{table} \\ 
	This tabular representation is used to describe the Bell inequality in Table \ref{BI_CGLMP3}. For the generalised $[m,d]$ scenario, the reordered hyperplane vector reads
	\begin{table}[h]
		\centering
		\begin{tabular}{p{1cm}  p{1cm} p{1cm} | p{1cm}  p{1cm} p{1cm} | p{1cm}  p{1cm} p{1cm} }
			$h_{A_1B_1}^{11}$ & $ \dots $ & $h_{A_1B_1}^{1d}$ &  & $\dots$ &  & $h_{A_1B_m}^{11}$ & $ \dots $ & $h_{A_1B_m}^{1d}$ \\  
			$ \vdots $ & $ \ddots $ & $ \vdots $ &$ \vdots $ & $ \ddots $ & $ \vdots $ & $ \vdots $ & $ \ddots $ & $ \vdots $ \\ 
			$h_{A_1B_1}^{d1}$ & $ \dots $ & $h_{A_1B_1}^{dd}$ &  & $\dots$ &  & $h_{A_1B_m}^{d1}$ & $ \dots $ & $h_{A_1B_m}^{dd}$ \\ 
			\hline 
			   & $\dots$ &  &   & $\dots$ &  &   & $\dots$ &  \\
			 $ \vdots $ & $ \ddots $ & $ \vdots $ &$ \vdots $ & $ \ddots $ & $ \vdots $ & $ \vdots $ & $ \ddots $ & $ \vdots $ \\  
			   & $\dots$ &  &   & $\dots$ &  &   & $\dots$ &  \\    
			\hline
			$h_{A_mB_1}^{11}$ & $ \dots $ & $h_{A_mB_1}^{1d}$ &  & $\dots$ &  & $h_{A_mB_m}^{11}$ & $ \dots $ & $h_{A_mB_m}^{1d}$ \\  
			$ \vdots $ & $ \ddots $ & $ \vdots $ & $ \vdots $ & $\ddots$ & $ \vdots $ & $ \vdots $ &  $ \ddots $  & $ \vdots $ \\ 
			$h_{A_mB_1}^{d1}$ & $ \dots $ & $h_{A_mB_1}^{dd}$ &  & $\dots$ &  & $h_{A_mB_m}^{d1}$ & $ \dots $ & $h_{A_mB_m}^{dd}$ \\ 
		\end{tabular}
		\caption{Bell inequality table for the $[m,d]$ scenario.}
		\label{CGLMP hyperplane table [m,d]}
	\end{table}	
	\twocolumngrid
	\bibliographystyle{ieeetr}
	\bibliography{references}
	
\end{document}